\documentclass[11pt]{article}

\usepackage{amsthm}
\usepackage{algorithm,algorithmic}
\usepackage{graphicx}
\usepackage{amssymb}

\def\denseformat{
	\setlength{\textheight}{9in}
	\setlength{\textwidth}{6.9in}
	\setlength{\evensidemargin}{-0.2in}
	\setlength{\oddsidemargin}{-0.2in}
	\setlength{\headsep}{10pt}
	\setlength{\topmargin}{-0.3in}
	\setlength{\columnsep}{0.375in}
	\setlength{\itemsep}{0pt}
}

\newtheorem{thm}{Theorem}[section]
\theoremstyle{definition}

\theoremstyle{plain}
\newtheorem{lem}[thm]{Lemma}
\newtheorem{col}[thm]{Corollary}

\newtheorem{dfn}[thm]{Definition}

\denseformat

\begin{document}

\title{Locally-Iterative Distributed ($\Delta + 1$)-Coloring below Szegedy-Vishwanathan Barrier, and Applications to Self-Stabilization and to Restricted-Bandwidth Models}



\author{Leonid Barenboim\thanks{Open University of Israel.
		E-mail: {\tt leonidb@openu.ac.il}
		\newline  $^{**}$ \   Ben-Gurion University of the Negev. Email: {\tt elkinm@cs.bgu.ac.il}
		\newline $^{***}$ Open University of Israel. Email: {\tt uri.goldenberg@gmail.com} \newline This research has been supported by Israel Science Foundation grant 724/15.}
	\and Michael Elkin$^{**}$ \and Uri Goldenberg$^{***}$}






\begin{titlepage}
	\def\thepage{}
	\maketitle
	
\begin{abstract}
	We consider graph coloring and related problems in the distributed message-passing model. {\em Locally-iterative algorithms} are especially important in this setting. These are algorithms in which each vertex decides about its next color only as a function of the current colors in its $1 - hop-neighborhood$. In STOC'93 Szegedy and Vishwanathan showed that any locally-iterative $(\Delta + 1)$-coloring algorithm requires $\Omega(\Delta \log \Delta + \log^* n)$ rounds, unless there exists "a very special type of coloring that can be very efficiently reduced" \cite{SV93}. 
	No such special coloring has been found since then. This led researchers to believe that Szegedy-Vishwanathan barrier is an inherent limitation for locally-iterative algorithms, and to explore other approaches to the coloring problem \cite{BE09,K09,B15,FHK16}.
	The latter gave rise to faster algorithms, but their heavy machinery which is of non-locally-iterative nature made them far less suitable to various settings. In this paper we obtain the aforementioned special type of coloring. Specifically, we devise a locally-iterative $(\Delta + 1)$-coloring algorithm with running time $O(\Delta + \log^* n)$, i.e., {\em below} Szegedy-Vishwanathan barrier. This demonstrates that this barrier is not an inherent limitation for locally-iterative algorithms. As a result, we also achieve significant improvements for dynamic, self-stabilizing and bandwidth-restricted settings. This includes the following results.
	
	\begin{itemize}
		\item We obtain self-stabilizing distributed algorithms for $(\Delta + 1)$-vertex-coloring, $(2\Delta - 1)$-edge-coloring, maximal independent set and maximal matching with $O(\Delta + \log^* n)$ time. This significantly improves previously-known results that have $O(n)$ or larger running times \cite{GK10}.
		
		\item We devise a $(2\Delta - 1)$-edge-coloring algorithm in the CONGEST model with $O(\Delta + \log^* n)$ time and $O(\Delta)$-edge-coloring in the Bit-Round model with $O(\Delta + \log n)$ time. The factors of $\log^* n$ and $\log n$ are unavoidable in the CONGEST and Bit-Round models, respectively. Previously-known algorithms had superlinear dependency on $\Delta$ for $(2\Delta - 1)$-edge-coloring in these models.
		
		\item We obtain an arbdefective coloring algorithm with running time $O(\sqrt \Delta + \log^* n)$. Such a coloring is not necessarily proper, but has certain helpful properties. We employ it in order to compute a proper $(1 + \epsilon)\Delta$-coloring within $O(\sqrt \Delta + \log^* n)$ time, and $(\Delta + 1)$-coloring within $O(\sqrt {\Delta \log \Delta} \log^* \Delta + \log^* n)$ time. This improves the recent state-of-the-art bounds of Barenboim from PODC'15 \cite{B15} and Fraigniaud et al. from FOCS'16 \cite{FHK16} by polylogarithmic factors.
		
		\item Our algorithms are applicable to the SET-LOCAL model \cite{HKMS15} (also known as the weak LOCAL model). In this model a relatively strong lower bound of $\Omega(\Delta^{1/3})$ is known for $(\Delta + 1)$-coloring. However, most of the coloring algorithms do not work in this model. (In \cite{HKMS15} only Linial's $O(\Delta^2)$-time algorithm and Kuhn-Wattenhofer $O(\Delta \log \Delta)$-time algorithms are shown to work in it.) We obtain the first linear-in-$\Delta$ $(\Delta + 1)$-coloring algorithms that work also in this model.
		
	\end{itemize}
	
\end{abstract}
\end{titlepage}

\section{Introduction}
{\large \bf 1.1 The Classical Model\\}
In the {\em LOCAL model} of distributed computing \cite{L87} a network is represented by an $n$-vertex graph $G = (V,E)$ with maximum degree $\Delta$ whose vertices host processors. The vertices communicate with one another over the edges of $G$ in {\em synchronous}  rounds. In each round vertices perform local computations and exchange messages with their neighbors. The amount of local computations, as well as message size, is unrestricted. The {\em running time} is the number of rounds from the beginning of the execution until all vertices compute their respective parts in the solution. Another model of interest is the {\em CONGEST model}, which is similar to the LOCAL model, except that message size is restricted to $O(\log n)$ bits per edge per round.

 The problem that we are studying is how many rounds 
are required for computing a proper\footnote[1]{A coloring $\varphi : V \rightarrow [\Delta + 1]$ is called {\em proper}, if $\varphi(u) \neq \varphi(v)$, for every edge $e = (u,v) \in E$.}  $(\Delta + 1)$-coloring of $G$. This is one of the most fundamental and well-studied distributed symmetry-breaking problems \cite{CV86,GP87,L87,SV93,KW06,BE10,BE11,BE13,BEK14,BEPS12,B15,FHK16}, and it has numerous applications to resource and channel allocation, scheduling, workload balancing, and to mutual exclusion \cite{K09,GK10}.
The study of distributed coloring algorithms on paths and cycles was initiated by Cole and Vishkin in 1986 \cite{CV86}, who devised a $3$-coloring algorithm with $O(\log^* n)$ time\footnote[2]{Unless said otherwise, algorithms that we discuss are deterministic.}. The first distributed algorithm for the $(\Delta + 1)$-coloring problem on general graphs was devised by Goldberg and Plotkin in 1987 \cite{GP87}. The running time of their algorithm is $2^{O(\Delta)} + O(\log^* n)$. ($\log^*$ is a very slow-growing function, defined formally in Section \ref{sc:pr}.) Goldberg, Plotkin and Shannon \cite{GPS88} improved this bound to $O(\Delta^2 + \log^* n)$. Linial \cite{L87} showed a lower bound of $\frac{1}{2} \log^* n - O(1)$. His lower bound applies to a more relaxed $f(\Delta)$-coloring problem, for any, possibly quickly-growing function $f()$. Linial also strengthened the upper bound of \cite{GPS88}, and showed that an $O(\Delta^2)$-coloring can be computed in $\log^* n + O(1)$ time. (Via a standard color reduction, described e.g., in \cite{BE13} Chapter 3, given an $\alpha$-coloring one can compute a $(\Delta + 1)$-coloring in $\alpha - (\Delta + 1)$ rounds. Thus, Linial's algorithm also gives rise to $(\Delta + 1)$-coloring in $O(\Delta^2 + \log^* n)$ time.)

In STOC'93, Szegedy and Vishwanathan \cite{SV93} studied {\em locally-iterative} coloring algorithms. An algorithm $\mathcal{A}$ is an $\alpha$-to-$\beta$ locally-iterative, for a pair of parameters $\alpha > \beta$, if it maintains a sequence $\varphi_1, \varphi_2,...,\varphi_T$ of {\em proper} $\alpha$-colorings, where $\varphi_i$ is the coloring on round $i$, for every $1 \leq i \leq T$, the coloring $\varphi_T$ is a $\beta$-coloring, and $T$ is the running time of the algorithm. On each round $i$, every vertex $v$ computes its new color $\varphi_{i+1}(v)$ based only on the colors $\{\varphi_i(u) \ | \ u \in \hat{\Gamma}(v)\}$, where $\hat{\Gamma}(v) = \{v\} \cup \{u \in V \ | \ (u,v) \in E \}$ is the $1 - hop-neighborhood$ of $v$.
Szegedy and Vishwanathan \cite{SV93} derived an improved upper bound of $O(\Delta \log \Delta + \log^* n)$ for locally-iterative $(\Delta + 1)$-coloring. Specifically, they devised an $O(\Delta^2)$-to-$(\Delta + 1)$-locally-iterative algorithm with running time $O(\Delta \log \Delta)$. (This upper bound was later re-derived in a somewhat more explicit way by Kuhn and Wattenhofer \cite{KW06}.) 
Szegedy and Vishwanathan \cite{SV93} also showed a {\em heuristic} lower bound on the number of rounds that a locally-iterative algorithm needs in order to compute a $(\Delta + 1)$-coloring from an $O(\Delta^2)$-coloring. Their lower bound (Theorem 12 in \cite{SV93}, marked as "heuristic") is $\Omega(\Delta \log \Delta)$. 
By Linial's lower bound \cite{L87}, $\frac{1}{2} \log^* n - O(1)$ rounds are required to compute an $O(\Delta^2)$-coloring.

All $(\Delta + 1)$-coloring algorithms developed before 2009 were locally iterative. (See Table 1 below for a summary of known locally-iterative algorithms.) 
However, since 2009, a variety of algorithms that employ various complicated non-locally-iterative techniques were obtained. This started with the linear-in-Delta algorithms of \cite{BE09,K09,BEK14}, and proceeded with sublinear algorithms of \cite{B15,FHK16}.
The algorithms of \cite{BE09,K09,B15,FHK16} are all not locally-iterative, as they all decompose the graph into many subgraphs, compute colorings for them, and carefully  combine them into a single coloring for the original graph. In view of Szegedy-Vishwanathan's heuristic lower bound (henceforth, {\em SV barrier}), this seemed to be inevitable.
In the current paper we show that this is not the case, and devise the first {\em locally-iterative} $(\Delta + 1)$-coloring algorithm with running time $O(\Delta + \log^* n)$, i.e., {\em below the SV barrier} of $\Omega(\Delta \log \Delta + \log^* n)$. Unlike previously locally-iterative algorithms, our algorithm does not necessarily reduce the number of employed colors in every round. Instead, if the initial number of colors is $\Delta^2$, it can keep being $\Omega(\Delta^2)$ for almost the entire execution of the algorithm, and then ''suddenly" reduce to $\Delta + 1$ in the last few rounds. The colorings $\varphi_1, \varphi_2, ...., \varphi_T$, $T = O(\Delta)$, that it computes on rounds $1,2,...,T$, respectively, are all proper, but they are {\em not at all arbitrary}. Rather they have some special properties that guarantee that in $O(\Delta)$ rounds the number of colors reduces to $(\Delta + 1)$.

Interestingly, in their seminar paper \cite{SV93}, Szegedy and Vishwanathan mention a possibility of such a phenomenon. In the preamble to their aforementioned "heuristic" theorem (Theorem 12) they wrote: 

\textit{''There is a possibility, however, that after a few steps of iteration we arrive at a very special type of coloring that can be very efficiently reduced in steps thereafter. Assuming that this does not happen, the results of the previous section give the following theorem:\\
	Theorem 12 (heuristic): Let $1 \leq b < a \leq \Delta/2$. To decrease the number of colors from $a\Delta$ to $b\Delta$ it takes $\Theta(\Delta \log(a/b))$ steps. In particular, to decrease the number
	of colors from $\Delta^2/2$ to $\Delta$ requires $\Theta(\Delta \log \Delta)$ steps."}\footnote[1]{The argument of \cite{SV93} applies, in fact, to reducing the number of colors to $\Delta + 1$, as opposed to $\Delta$.}

We also use our new locally iterative technique to devise improved {\em not} locally-iterative coloring algorithms. Specifically, we obtain $(1 + \epsilon)\Delta$-coloring within $O(\sqrt{\Delta} + \log^* n)$ time, for an arbitrarily small constant $\epsilon > 0$, and a $(\Delta + 1)$-coloring within $O(\sqrt {\Delta \log \Delta} \log^* \Delta + \log^* n)$ time. This improves the best previously-known running time $O(\sqrt \Delta \log^{2.5} \Delta + \log^* n)$ of Fraigniaud et al. \cite{FHK16}, by a polylogarithmic in $\Delta$ factor. 

\begin{table}[H]
	\begin{center}
		
		\begin{tabular}{|| c c ||}
			\hline
			Running time & Reference \\ 
			\hline
			$2^{O(\Delta)} + O(\log^*n)$  & Goldberg, Plotkin \cite{GP87} \\ 
			$O(\Delta^2) + \log^*n$ & Linial \cite{L87} \\
			$O(\Delta)\cdot \log{n}$ & Goldberg at el. \cite{GPS88} \\
			$O(\Delta^2) + \log^*n$ & Goldberg et al. \cite{GPS88} \\ 
			$O(\Delta \log \Delta) + \frac{1}{2} \log^*n$  & Szegedy, Vishwanathan \cite{SV93}\\  
			$O(\Delta\log\Delta) + \log^*n$ & Kuhn, Wattenhofer \cite{KW06}\\
			$O(\Delta) + \log^*n$ & {\bf This paper}\\
			\hline
		\end{tabular}
	\end{center}
	\caption{Known results for locally-iterative $(\Delta + 1)$-coloring.}
\end{table}

{\noindent \large \bf 1.2 Our Locally-Iterative Algorithms\\}
We start with describing our most basic subroutine, which we call {\em Additive Group algorithm}, or shortly, AG algorithm. The subroutine starts with a proper $(\Delta + 1)^2$-vertex-coloring $\varphi$ of the input graph $G$, and produces its proper $(\Delta + 1)$-coloring in $O(\Delta)$ rounds, in a locally-iterative way. Assume (for simplicity of presentation) that $\Delta + 1 =p$ is a prime number. We represent every initial color $\varphi(v) = \varphi_0(v)$ as a pair $\langle a_v, b_v \rangle$, where $a_v,b_v$ are from the field of integers with characteristic $p$, i.e.,  $a_v,b_v \in GF(p)$. Then every vertex $v \in V$ (in parallel) checks if there exists a neighbor $u \in \Gamma (v)$, with $b_u = b_v$. If there is no such a neighbor, then the vertex $v$ {\em finalizes} its color, i.e., sets it to $\langle 0,b_v \rangle$. Otherwise, the vertex $v$ sets its color to $\langle a_v, b_v + a_v \rangle$, where the addition is performed in $GF(p)$. We show (see Section \ref{sc:ag}) that when all vertices run this simple iterative step for $2p + 1 = 2(\Delta + 1) + 1$ rounds, the ultimate coloring $\psi$ is a proper $(\Delta + 1)$-coloring. Moreover, at all times the graph is properly colored.

The simplicity and the uniformity of this iterative step makes it very powerful. In dynamic self-stabilizing environments vertices run this step forever in conjunction with an appropriate ''check-and-fix" procedure, no matter what changes or faults occur in the network. It turns out that still, once faults stop occurring, within additional $O(\Delta)$ rounds the coloring converges to a proper $(\Delta + 1)$-coloring. In the edge-coloring scenario, every edge $e = (u,v)$ has a color $\varphi(e) = \langle a_e, b_e \rangle$, known to both endpoints. The endpoint $u$ checks locally if there is an edge $e_u$ incident on $u$, $e_u \neq e$, with $b_{e_u} = b_e$, and $v$ makes an analogous test among edges incident on it. Then $u$ and $v$ communicate to one another {\em one single bit} each, which enables both of them to update the color of $e$. Therefore, this algorithm gives rise to the first communication- and time-efficient $(2\Delta - 1)$-edge-coloring algorithm. 

Some subtleties arise when $(\Delta + 1)$ is not prime, and we overcome them by showing that in some cases the proof goes through even if the arithmetics is performed in an additive group $Z_{\Delta + 1}$, rather than in a Galois field $GF(p)$. Another difficulty stems from the need to combine the AG algorithm with Linial's algorithm. The latter algorithm reduces the number of colors to $O(\Delta^2)$, and from there the AG algorithm takes over. However, in the self-stabilizing setting some vertices may run Linial's algorithm, while others have already proceeded to AG algorithm. Careful adaptations to both algorithms are required to handle such situations.

Finally, we also extend the AG algorithm to computing {\em arbdefective coloring}. For a pair of parameters  $\alpha$ and $\beta$, a coloring $\varphi$ is said to be {\em $\alpha$-arbdefective $\beta$-coloring} if the $\beta$ color classes of $G$ induce subgraphs of arboricity at most $\alpha$ each. Arbdefective colorings were introduced by the first- and the second-named authors in \cite{BE10}, and they were shown to be extremely useful for efficient computation of proper colorings in \cite{BE10,B15,FHK16}. Our extension of AG algorithm from proper to arbdefective colorings (we call the extended algorithm {\em ArbAG}) works very similarly to the AG algorithm. The only difference is that on each round, each vertex $v$ tests if it has at most a certain number of neighbors $u$ with $b_u = b_v$. (Recall that in AG algorithm, this threshold number is $0$.) Other than that ArbAG has the same simple locally-iterative structure as algorithm AG, but the number of iterations of ArbAG is significantly smaller. (Note, however, that strictly speaking, a locally iterative algorithm is required to maintain a proper coloring on each round, while algorithm ArbAG maintains an arbdefective coloring.)
This is in sharp contrast to previous methods \cite{BE10,B15} of computing arbdefective colorings. The latter are far more involved, far less communication-efficient, and less time-efficient by polylogarithmic factors.
As a result we also obtain improved (again, by polylogarithmic factors) algorithms for general (not necessarily locally-iterative) $(\Delta + 1)$-coloring and $(1 + \epsilon)\Delta$-coloring. 


{\noindent \large \bf 1.3 Applications\\}
In the Conclusions section of the paper \cite{KW06} by Kuhn and Wattenhofer, the authors explain why locally-iterative algorithms are particularly important from practical perspective. They mention ''emerging dynamic and mobile distributed systems such as peer-to-peer, ad-hoc, or sensor networks" as examples of networks for which such algorithms can be especially suitable. They also point out that locally-iterative algorithms are typically communication-efficient ones. 

In this paper we demonstrate that our novel locally-iterative algorithms indeed provide dramatically improved bounds for both the dynamic Self-Stabilizing scenarios and for scenarios in which communication-efficiency is crucial. In the next three subsections we discuss these applications of our locally-iterative technique one after another.

{\noindent \bf 1.3.1 Self-Stabilizing Symmetry Breaking\\}
The Self-Stabilizing setting was introduced by Dijkstra \cite{D74}, and is being intensively studied since then. See, e.g., Dolev's monograph \cite{D00} and surveys by Herman \cite{H02}, by Guelleti and Kheddouci \cite{GK10}. The latter article surveys results on self-stabilizing coloring, independent sets and matchings that were achieved before 2010. Since then, several additional results were obtained, either for more general or more restricted scenarios. This includes distance-2 coloring of vertices \cite{BM12} and of edges \cite{LL14}, and maximal independent sets in wireless sensor networks \cite{AAD19}. Self-stabilization in dynamic systems was defined in \cite{DH97}. 

In the context of $(\Delta + 1)$-coloring, the setting we consider is the following one. The network is represented by a synchronous message-passing system with a synchronous scheduler and a distributed demon. Every vertex $v$ of a graph $G = (V,E)$ of maximum degree at most $\Delta$ and at most $n$ vertices has a unique ID number. In each round each vertex reads all messages that were received on its edges, produce new messages, performs local computations, and clears the memory used for storing the messages in the end of the round. The memory of each vertex consists of two areas. The {\em Read Only Memory} (henceforth, ROM) consists of hard-wired data such as vertex ID, degree bound $\Delta$, vertices bound $n$, and program code. The ROM is faultless, but its contents cannot be changed during execution. The other area of the memory is {\em Random Access Memory} (henceforth, RAM). This memory may change during execution, and it is appropriate for storing variables, such as vertex colors. 

The RAM area, however, may change not only as a result of an algorithm instruction, but also as a result of faults or adversarial activity of the demon. Since the demon is a distributed one, it may change the memories of numerous processors simultaneously. Such faults may make arbitrary and completely unpredictable changes in any round in the entire RAM in all vertices. In particular, the memory areas that store incoming and outgoing messages may be affected, thus messages may be lost or corrupted. Moreover, in the {\em Fully-Dynamic Self-Stabilizing setting}, in each round vertices may crash, new vertices may appear and communication links between vertices may change arbitrarily, as long as the bounds on $n$ and $\Delta$ hold\footnote[1]{In fact, since the dependence of our algorithms' running time on $n$ is just $\log^* n$, the bound for the number of vertices may be double- or triple-exponential in the real number of vertices, and still the running time will be affected by just an additive constant term.}.  For example, colors are stored in RAM, and as long as faults occur, vertices may hold arbitrary colors, possibly the same as those of their neighbors, no matter what operations are performed by an algorithm. 
The objective is to devise algorithms in which once faults and dynamic changes stop occurring, the algorithm {\em self-stabilizes} quickly to a proper solution.

The relevant notion of running time in this context is called {\em stabilization time} (also known as ''quiescence" time), which is the maximum number $T$ of rounds, so that $T$ rounds after the last fault or dynamic change of the graph we are guaranteed that an algorithm arrives to a proper solution, e.g., the coloring of the graph is a proper $(\Delta + 1)$-coloring. One can define analogously self-stabilizing variants of $(2\Delta - 1)$-edge-coloring (see Section 1.2.2), of Maximal Independent Set (henceforth, MIS) and of Maximal Matching (henceforth, MM)\footnote[2]{A subset $U \subseteq V$ of vertices is an {\em MIS} if there are no edges between pairs of vertices in $U$, and for every vertex $v \in V \setminus U$, there exists a neighbor $u \in U$. A subset $M \subseteq E$ of edges is an {\em MM} if no two edges of $M$ are incident, and for every $e \in E \setminus M$, there exists an edge $e' \in M$ incident on it.}.

Self-stabilizing symmetry-breaking problems were extensively studied \cite{HH92,IKK02,KK06,SS93}. See also \cite{GK10} for an excellent survey of self-stabilizing symmetry-breaking algorithms. However, all of them have prohibitively large stabilization time of $O(n)$ or more. A general scheme for transforming $T$-round algorithms from the LOCAL model into $T$-round self-stabilizing  algorithms was described in \cite{LSW09}. This, however, may result in a significant growth in the message size, due to the need of collection information of $T$-hop-neighborhoods. In contrast, in this paper we devise the first self-stabilizing algorithms with stabilization time of $O(\Delta + \log^* n)$ and small messages, for all these four fundamental problems.
We note that the fact that our algorithms are {\em deterministic} is particularly useful in this setting. Indeed, this prevents the possibility that adversarial faults will manipulate random bits of the algorithm.

{\noindent \bf 1.3.2 Edge-Coloring\\}
Another classical and extremely well-studied symmetry breaking problem is that of $(2\Delta - 1)$-edge-coloring \cite{PR01,BE11,BEM17,BEPS12,EPS15,DGP98,GP97,FGK17,PS97}. An {\em edge-coloring} $\varphi$ of a graph $G = (V,E)$ is a function $\varphi: E \rightarrow N$. It is said to be {\em proper} if for every pair of incident edges $e,e' \in E$, $e \neq e'$, we have $\varphi(e) \neq \varphi(e')$. The classical theorem of Vizing \cite{V64} states that every graph is $(\Delta + 1)$-edge-colorable. However, existing distributed deterministic solutions \cite{PR01,BE11,BEM17,BEPS12,EPS15,FGK17} with running time of the form $f(\Delta) + O(\log^* n)$ employ $(2\Delta - 1)$ colors or more in general graphs. 
(There are efficient randomized distributed algorithms \cite{BEPS12,EPS15} that compute $(1 + \epsilon)\Delta$-edge-colorings in time close to $(\log n)/ \Delta^{1 - o(1)}$. This running time is incomparable to running time of the form $f(\Delta) + O(\log^* n)$, for some function $f()$, achieved by deterministic algorithms that we discuss here.)
The first efficient deterministic algorithm for $(2\Delta - 1)$-edge-coloring was devised by Panconesi and Rizzi \cite{PR01}. Its running time is $O(\Delta + \log^* n)$.

In the LOCAL model of distributed computing, messages of arbitrary size are allowed. The $(2\Delta - 1)$-edge-coloring problem for a graph $G$ reduces to $(\Delta + 1)$-vertex-coloring problem for the line graph $L(G)$ of $G$, and in the LOCAL model this reduction can be implemented without any overhead in running time. Therefore, the novel sublinear-in-$\Delta$ time algorithms for $(\Delta + 1)$-vertex-coloring \cite{B15,FHK16} immediately give rise to sublinear-in-$\Delta$ time algorithms for $(2\Delta - 1)$-edge-coloring. However, all these edge-coloring algorithms \cite{PR01,B15,HKMS15} are not locally iterative. Moreover, they do not apply (or require significantly more time) in the CONGEST model of distributed computing. 
Implementing Panconesi-Rizzi algorithm in the CONGEST model requires $O(\Delta^2 + \log^* n)$ time. Simulating vertex-coloring for a line graph also yields a multiplicative overhead of factor at least $\Delta$ in the running time. Therefore, to the best of our understanding, the state-of-the-art solution for $(2\Delta - 1)$-edge-coloring in the CONGEST model requires $\tilde{O}(\Delta^{3/2} + \log^* n)$ time, and it is not locally iterative. The best currently-known locally-iterative solution is even slower, and requires $O(\Delta^2 \log \Delta + \log^* n)$ time.  (It is achieved by simulating the locally-iterative $O(\Delta \log \Delta)$-time algorithm of \cite{KW06,SV93} in the line graph in the CONGEST model.)
The problem of devising communication-efficient algorithms for symmetry-breaking problems was raised in a recent work by Pai et al. \cite{PPPRR17}.

We adapt our locally-iterative algorithm for $(\Delta + 1)$-vertex-coloring to work for $(2\Delta - 1)$-edge-coloring directly, i.e., without simulation of the line graph. As a result we obtain a locally-iterative $(2\Delta - 1)$-edge-coloring algorithm with running time $O(\Delta + \log^* n)$ in the CONGEST model. Moreover, we show that unlike previous solutions (that require stabilization time of $\Omega(n)$), our algorithm works in the self-stabilizing setting, still with small messages, with stabilization time $O(\Delta + \log^* n)$. Moreover, our algorithm is also applicable to the more restricted Bit-Round \cite{KSOS06} model in which each vertex is only allowed to send 1 bit in each round over each edge. 

As a separate contribution, we devise a $(2\Delta - 1)$-edge-coloring algorithm for $n$-vertex oriented forests that requires $\log^* n + O(1)$ time, and applies to the CONGEST model. The currently existing solution to this problem that has this running time, due to Panconesi and Rizzi \cite{PR01}, employs messages of size $O(\Delta)$.
\\

{\noindent \bf 1.3.3 SET-LOCAL Model\\}
An additional application of our algorithms is in the SET-LOCAL model \cite{HKMS15} that represents restricted networks in which vertices do not have IDs (but start from a proper coloring), and are not capable to distinguish between identical messages received from different neighbors. Since our algorithms are locally-iterative and compute the next colors based only on sets of current colors of $1$-hop-neighborhoods, our algorithms are directly applicable to the SET-LOCAL model. Thus our algorithms compute proper $(\Delta + 1)$-coloring (and solve related problems) in $O(\Delta)$ time in the SET-LOCAL model starting from a proper $O(\Delta^2)$ coloring. The best previous algorithms in this model required $O(\Delta \log \Delta)$ time \cite{SV93,KW06,HKMS15}. 
A lower bound of $\Omega(\Delta^{1/3})$ for $(\Delta + 1)$-coloring in this setting was obtained by Hefetz et al. \cite{HKMS15}.

{\noindent \large \bf 1.3.4 Summary\\}
We believe that these applications demonstrate the power of locally-iterative coloring. Bypassing Szegedy-Vishwanathan barrier via a locally-iterative algorithm does not only provide a surprising answer to a quarter-century-old open problem, but also provides new precious insights into distributed coloring in general. We are confident that these insights will be instrumental in achieving further breakthroughs in this important field.\\

\section{Preliminaries} \label{sc:pr}
The function $\log^* n$ is the number of times the $\log_2$ function has to be applied iteratively starting from $n$, until we arrive at a number smaller than $2$.
The unique identity number (ID) of a vertex $v$ in a graph $G$ is denoted $id(v)$.
The diameter $Diam(G)$ of a graph $G = (V,E)$ is the maximum (unweighted) distance between vertices $u,v \in V$. The {\em arboricity} $a = a(G)$ of a graph $G = (V,E)$ is the minimum number of forests into which the edge set $E$ can be partitioned.
A {\em $d$-defective $p$-coloring} is a vertex coloring using $p$ colors such that each vertex has at most $d$ neighbors colored by its color. A {\em $b$-arbdefective $p$-coloring} is a vertex coloring using $p$ colors, such that each subgraph induced by vertices of the same color has arboricity at most $b$.
We employ the following important fact. For any integer $\Delta > 0$, there exists a prime $q$ in $[\Delta, 2\Delta]$.
This is due to Bertrand-Chebyshev postulate. See, e.g., Theorem 418 in \cite{hardy_write}.

\section{Additive-Group Coloring} \label{sc:ag}
\subsection{The Main Algorithm}
In this section we present our main algorithm that computes a proper $O(\sqrt{k})$-coloring from a proper $k$-coloring, where $k =\Omega(\Delta^2)$.
Consider a graph $G=(V, E)$ with a proper $k$-coloring $\psi$. For all vertices $v \in V$, we represent a color $\psi(v) = i$ by a pair $
\langle a_v,b_v \rangle$. We do it by finding a prime number $q, \sqrt {k} \leq q \leq 2\sqrt{k} $. 
The color $\psi(v) = i$ is represented by the following pair $\psi(v) = \langle \left \lfloor i/q \right \rfloor, i \bmod q \rangle$. 
Our final goal is to eliminate the first coordinate, i.e., to change all nodes colors such that for every vertex $v \in V$, it will hold that $\psi(v) = \langle 0, b_v\rangle $, $0 \leq b_v < q$, and $\psi$ is a proper $q$-coloring. Our algorithm proceeds in iterations, starting from the initial coloring $\psi$. In each iteration colors may change, but the coloring remains proper. We employ the following definition.

\begin{dfn} Two neighbors $u,v$ in $G$ {\em conflict} with one another if and only if $\psi(v) = \langle a, b\rangle $ and $\psi(u) = \langle a', b\rangle $, where $0 \leq a,b,a' < q$.
\end{dfn}

Denote $\psi(v) = \langle a, b\rangle $. We will refer to $a$ as the first coordinate and to $b$ as the second coordinate.
Denote by $\psi_i(v)$ the color of $v \in V$ in round $i$.
Our algorithm starts from a proper $k = \Omega(\Delta^2)$ coloring of the input graph $G = (V,E)$. In each round the algorithm performs the following step, for $q$ rounds. For all $v \in V$ in parallel, if a node $v$ conflicts with a neighboring node $u$, then the new color of $v$ in the end of this round is $\psi_{i+1}(v) = \langle a, (b + a) \bmod q\rangle $. Otherwise (this means $v$ does not conflict with any neighbor), we set $\psi_{i+1}(v) = \langle 0, b\rangle $, and the color of $v$ becomes final and will not change anymore.\footnote[1]{Note, however, that a {\em finalized} vertex $v$, i.e., a vertex with $\psi_i(v) = \langle 0, b \rangle$, can keep running the same iterative step, and still its colors will stay unchanged.} This completes the description of the algorithm. Note that a node does not have to send its new color to all of its neighbors. Rather it is enough to send only one bit indicating whether its color became final or that it changed according to the rule specified above. We will use this property later. The pseudocode of the algorithm is provided below. (The pseudocode is for a specific vertex $v$ that runs this algorithm. All vertices run it in parallel.) Next, we prove correctness.

\clearpage

\begin{algorithm}
	\caption{Additive-Group Coloring}
	\begin{algorithmic}[1]
	
	  \STATE /* Initially, each vertex is aware of its own color and the colors of its neighbors */
		
		\FOR {round $i = 0,1,,... q$}
		
		\STATE let $\psi_i(v) = \langle a_v , b_v \rangle$ be the color of $v$ in iteration $i$ 
		
		\IF {not exists $ (v , u) \in E$ where $\psi_i(u) = \langle a_u , b_u \rangle$ with $b_u = b_v$ }
		
		\STATE $\psi_{i+1}(v) = \langle 0  , b_v \rangle$
		
		\STATE Send $0$ to all neighbors
		
		\ELSE

		\STATE $\psi_{i +1}(v) =  \langle a_v , (b_v + a_v)  \mbox{ mod } q \rangle $		
		
		\STATE Send $1$ to all neighbors
		
		\ENDIF
		
		\STATE Receive the bits sent by neighbors of $v$ and deduce the colors $\psi_{i+1}$ of these neighbors
		
		\ENDFOR
	\end{algorithmic}
\end{algorithm}

\begin{lem} \label{lem:propercol}
	For each iteration $i$, the coloring $\psi_i(G)$ is proper.
\end{lem}
\begin{proof}
	The proof is by induction on $i$.\\
	{\bf Base: ($i = 0$)}:  holds trivially, since the initial coloring is proper.\\
	{\bf Step:} Assuming that in iteration $i$ the coloring is proper, we prove that in iteration $i + 1$ it is proper as well. If a color  of a node $v \in V$ is $\psi_i(v) = \langle a, b\rangle $, then for the next iteration the color is either $\psi_{i+1}(v) = \langle 0, b\rangle $ or $\psi_{i+1}(v) = \langle a, (b + a) \bmod q\rangle $.
	Consider an adjacent node $u$, i.e., $(u,v) \in E$. If $\psi_i(u) = \langle c,b\rangle $, where $0 \leq c < q$, then $c \ne a$, by the induction hypothesis. In this case, the new colors of the nodes will be $\psi_{i+1}(v) = \langle a, (b + a) \bmod q\rangle $ and $\psi_{i+1}(u) = \langle c, (b + c) \bmod q\rangle $ and since $c \ne a$ this means that the new colors of $u$ and $v$ are distinct. Otherwise, $\psi_i(u) = \langle c,d\rangle $, where $d \neq b$. If in iteration $i + 1$ it holds that $\psi_{i+1}(v) = \langle 0,b\rangle $ and $\psi_{i+1}(u) = \langle 0,d\rangle $, we are done since $b \neq d$. Otherwise, $u$ or $v$ had conflicts in iteration $i$. If exactly one of them had a conflict, then their colors in iteration $i + 1$ are distinct. (One of them has 0 in the first coordinate, while the other has not, in iteration $i + 1$.) It is left to consider the case that both had conflicts. Thus, $\psi_{i + 1}(v) = \langle a, (b + a) \bmod q\rangle $ and $\psi_{i+1}(u) = \langle c, (d + c) \bmod q\rangle $. If $a \neq c$, we are done. Otherwise, $a = c$ and $b \neq d$, because $\psi_i$ is proper. Thus, $b + a  \not\equiv  
	d + c \mbox{ } (\bmod \ q)$, and $\psi_{i+1}(v) \neq \psi_{i+1}(u)$.
\end{proof}
We say that a vertex is in a {\em working} stage as long as its color $\langle a,b\rangle $ satisfies $a \neq 0$. Once $a$ becomes $0$, the vertex is in the {\em final} stage.
In order to analyze the running time of the algorithm we observe in Lemmas \ref{lem:enoughcolors}, \ref{lem:enoughcolorsb}  and Corollary \ref{col:squarerootcoloring}, assuming that $q$ is sufficiently large, that a pair of neighbors can conflict at most twice in $q$ rounds. (Once in a working stage, and once in a final stage of one of the vertices.) Therefore, a vertex with less than $q/2$ neighbors will have a round out of $q$ in which it conflicts with no neighbor. In this round it will select a final color. Since $q > 2 \cdot \Delta$, all vertices in the graph will select a color within $q$ rounds. This is formalized in the following analysis.
	\begin{lem} \label{lem:enoughcolors}
		For $t \leq q$, suppose that our algorithm is executed for $t$ rounds, and consider two neighboring nodes $u,v$ in $G$ that are in their respective working stages during these entire $t$ rounds.
		Then $u, v$ have the same second coordinate in their colors in the same round $i$, $0 \leq i < t$ (that is, $\psi_i(u) = \langle a,b\rangle $ and $\psi_i(v) = \langle c,b \rangle$, for some $0 \leq a,b,c < q$) at most once during these $t$ consequent rounds.
	\end{lem}
	\begin{proof}
		Assume that in some iteration $i$ it holds that $\psi_i(u) = \langle a,b\rangle $ and $\psi_i(v) = \langle c,b\rangle $.
		For each of the following iterations $j = i + 1, i + 2,...$, the difference between the second coordinates is $(c - a) \cdot (j - i) \bmod q$. 
		Note that since $q$ is a prime and $a \ne c$ (since, by Lemma \ref{lem:propercol}, the coloring is proper in all iterations, and in particular, $\psi_i$ is a proper coloring), the equality $(c - a)(j - i) \bmod q = 0$ can only hold  when $(j - i) \bmod q = 0$, i.e., only after additional $q$ iterations.
	\end{proof}
	In the following lemma we complement Lemma \ref{lem:enoughcolors}.
	\begin{lem} \label{lem:enoughcolorsb}
		For $t \leq q$, suppose that our algorithm is executed for $t$ rounds, and consider two neighboring nodes $u,v$ in $G$, such that $u$ is in working stage and $v$ is in final stage during these entire $t$ rounds.
		Then $u, v$ have the same second coordinate in their colors in the same round $i$, $0 \leq i < t$ (that is, $\psi_i(u) = \langle a,b\rangle $ and $\psi_i(v) = (0,b)$, for some $0 \leq a,b < q$) at most once during these $t$ consequent rounds.
	\end{lem}
	\begin{proof}
		Since $v$ is in final stage, its color does not change during these $t$ rounds. Indeed, it holds that $\langle 0,b\rangle  = \langle 0, (b + 0) \bmod q\rangle $. On the other hand, $u$ is in the working stage. If initially the color of $v$ is $\langle c,d\rangle $, for some $0 \leq c,d < q$, then in the following $t$ rounds it changes as follows: $\langle c, (d + c) \bmod q\rangle $, $\langle c, (d + 2c) \bmod q\rangle, \ldots,\langle c,(d + tc) \bmod q\rangle $. Since $q$ is prime, all these values of the second coordinate are distinct in the field of integers modulo $q$. In other words, the equality $d + xc  \equiv b \mbox{ } (\bmod \ q)$ holds for exactly one element $x$ of this field. Thus $v$ conflicts with $u$ at most once, in the round $i$ where $d + ic \equiv b \mbox{ } (\bmod \ q)$.
	\end{proof}
	\begin{col} \label{col:squarerootcoloring}
		Given a graph $G = (V, E)$ with a proper $k$-coloring, where $k = \Theta(\Delta^2)$, our Additive-Group Coloring algorithm produces a proper $O(\sqrt{k})$ coloring within $O(\Delta)$ rounds, each of which can be implemented via one-bit messages.
	\end{col}
	\begin{proof}
		By Lemma \ref{lem:enoughcolors}, for $q > 2\Delta$, two adjacent nodes in the working stage (whose colors are not final) cannot conflict with one other more than once during the first $q$ rounds of the algorithm.  However, two adjacent nodes can also conflict if exactly one of them has selected a final color. Once this happens, it will conflict with its neighbor that is still in the working stage at most once during these $q$ rounds. (See Lemma  \ref{lem:enoughcolorsb}.) Since any node starts from a working state, and once the state transits to final its color does not change anymore, a node cannot conflict with each of its neighbors more than twice. Therefore, for each node, within $q > 2 \cdot \Delta$ rounds, there must be a round in which it does not conflict with any of its neighbors. Hence, all nodes will reach a final stage within $q$ rounds. Since $q \leq 2\sqrt k = O(\Delta)$, the statement about the running time of the corollary follows. 
Recall also that on every round, each vertex $v$ can update its neighbors regarding its new color via one-bit messages. These messages indicate whether $v$ finalized its color or not.

		A final color is of the form $\langle 0,b\rangle $, $0 \leq b < q$. Thus the number of employed colors is at most  $q = O(\sqrt{k})$.
	\end{proof}
\begin{col} \label{col:propercoloring}
	Any graph $G = (V, E)$ can be colored with $\Delta + 1$ colors within $O(\Delta) + \log^* n$ rounds, by a locally-iterative algorithm.
\end{col}
\begin{proof} Running Linial's algorithm \cite{L87} on the input graph $G = (V, E)$ will produce a coloring $\varphi(G)$ using $O(\Delta^2)$ colors within $\log^*n + O(1)$ rounds. (Recall that Linial's algorithm is locally-iterative.)\\
	At the second stage we run our Additive-Group algorithm on $\varphi(G)$. This results in a new proper coloring $\psi(G)$ that employs $O(\Delta)$ colors. Computing the coloring $\psi$ from $\varphi$ requires $O(\Delta)$ rounds, by Corollary \ref{col:squarerootcoloring}.
	At the last stage we reduce the number of colors to $\Delta + 1$ using the standard color reduction. This also requires $O(\Delta)$ time. Note that the standard color reduction is a locally-iterative algorithm as well.
	Therefore, the overall running time is $\log^*n + O(1) + O(\Delta) + O(\Delta) = O(\Delta + \log^* n)$.
\end{proof}

\subsection{Halving the Number of Colors using 1-Bit-Messages per Round}
\label{sc:smbit}

In this section we devise a more bit-efficient algorithm than the algorithm presented in the previous section. Specifically, while the Additive-Group coloring stage requires just 1 bit per edge per round, the standard color reduction performed in the last stage may require $O(\log\Delta)$ bits for color updates for each round. We devise an improved method that requires messages of just $1$ bit. Specifically, we devise an algorithm reducing the number of colors from $O(\Delta^2)$ to $\Delta + 1$ within $O(\Delta)$ rounds using messages of $1$ bit per edge per round.
Consequently, the overall bit complexity of the $(\Delta + 1)$ coloring algorithm is $O(\log{n} + \Delta)$ in the one bit model.

In this algorithm there is no need for a prime parameter, but rather any integer greater than $\Delta$ will do. Given a graph with a proper $k$-coloring, $k \geq 2\Delta + 2$, we set $q = \lceil  \frac{k}{2} \rceil$, where $q \ge \Delta +1 $, and produce a proper $q$-coloring.
Initially, each color $c$, $0 \leq c < k$, is represented as an ordered pair:
$\langle \lfloor  c / q \rfloor, c \bmod q  \rangle$. Note that $\lfloor  c / q \rfloor \in \{0,1\}$. The pseudocode is provided below.

\begin{algorithm} 
	\caption{One-bit AG halving reduction} \label{alg:halvingalg}
	\begin{algorithmic}[1]
	
	  \STATE /* Initially, each vertex is aware of its own color and the colors of its neighbors */
		
		\FOR {round $i = 0,1,..., \Delta + 1$}
		
		\STATE let $\psi_i(v) = \langle a_v , b_v \rangle$ be the color of $v$ in iteration $i$ \ \ \ // $a_v \in \{0 , 1\}$
		
		\STATE $\forall v \in V$ such that $\psi_i(v) = \langle 1 , b_v \rangle$ in parallel do:
		
		\IF {not exists $ (v , u) \in E$ where $\psi_i(u) = \langle 0 , b_v \rangle $}
		
		\STATE $\psi_{i+1}(v) = \langle 0  , b_v \rangle$
		
		\STATE Send $0$ to all neighbors
		
		\ELSE

		\STATE $\psi_{i +1}(v) =  \langle 1 , b_v + 1  \mbox{ mod } q \rangle $		
		
		\STATE Send $1$ to all neighbors
		
		\ENDIF
		
		\STATE Receive the bits sent by neighbors of $v$ and deduce the colors $\psi_{i+1}$ of these neighbors
		
		\ENDFOR
	\end{algorithmic}
\end{algorithm}

We analyze the algorithm using the following lemmas.

\begin{lem} \label{lem:proper}
	Given an arbitrary graph $G = (V, E)$ with a proper $k \ge 2 \Delta +2$ coloring, one-bit AG halving reduction preserves a proper coloring of the input graph in every round.
\end{lem}
\begin{proof}
	Assume that in iteration $i$ the coloring is proper. Therefore, for every edge $(u,v) \in E$, we have $\langle a_u, b_u \rangle = \psi_i(u) \ne \psi_i(v) = \langle a_v, b_v \rangle$. In iteration $i + 1$ there are 2 possibilities.
	
	Case 1: $\psi_{i+1}(v) = \langle 0  , b_v \rangle$, and this means that $\psi_i(u) \ne \langle 0 , b_v \rangle$, since in this case $\psi_i(v)$ is either $\langle 0 , b_v \rangle$ or $\langle 1, b_v \rangle$. Moreover, this means that $\psi_{i+1}(u)$ cannot become $\langle 0, b_v \rangle$ during this iteration. Thus, $\psi_{i+1}(u) \neq \psi_{i + 1}(v)$.
	
	Case 2: $\psi_{i +1}(v) =  \langle 1 , b_v + 1  \mbox{ mod } q \rangle $. From the proper coloring assumption we know that if $\psi_{i}(u) =  \langle 1 , b_u \rangle $ then $\psi_i(v)$ is $\langle 1, b_v \rangle$ with $b_v \neq b_u$. Therefore, either $\psi_{i + 1}(u) = \langle 1, b_u + 1 \mbox{ mod } q \rangle \neq \psi_{i+1}(v)$ or $\psi_{i + 1}(u) = \langle 0, b_u \rangle \neq \psi_{i + 1}(v)$. On the other hand, if $\psi_i(u) = \langle 0, b_u\rangle$, then $\psi_{i+1}(u) = \langle 0, b_u \rangle$ as well, and again $\psi_{i+1}(u) \neq \psi_{i+1}(v)$.
\end{proof}

Next we show that Algorithm \ref{alg:halvingalg} actually halves the palette within $(\Delta + 1)$ rounds.

\begin{lem} \label{lem:enoughcolorsonebit}
	Given any graph $G = (V, E)$ with a proper $k \ge 2\Delta + 2$ coloring, One-bit AG halving reduction will cause every node to have a final color in the range $\{0,1,...,q-1\}, q = \lceil k/2 \rceil$, after $\Delta +1$ rounds.
\end{lem}
\begin{proof}
	Note that a node $u$ can conflict with another node $v$ in One-bit AG halving reduction if $\psi(u) = \langle 0 , b_v \rangle$ and $\psi(v) = \langle 1 , b_v \rangle$. After that these nodes may conflict again only once $q$ additional rounds have passed. Therefore, within $q$ rounds, a node can have a conflict at most once with every adjacent node. Thus, if $q \ge \Delta + 1$, from the pigeonhole principle there will always be a round where $v$  finalizes its color.   
\end{proof}

	Now we discuss the scenario when $(\Delta + 1)$-coloring is computed from scratch. To this end, Linial's algorithm is executed first. In each of its $O(\log^* n)$ rounds, vertices exchange their colors in that round with their neighbors. The ranges of colors in round $1,2,3,...$ are $O(n)$, $O(\Delta \log n)^2$, $O(\Delta \log \log n)^2$,..., respectively. Consequently, the bit complexity per edge is $O(\log n + \log \Delta + \log \log n + \log \Delta + \log \log \log n + ....) = O(\log n + \log \Delta \cdot \log^* n)$. Note that $\log \Delta \log^* n \leq O(\log n + \Delta)$. We summarize this in the next corollary.
	
	\begin{col} \label{col:graphd2col}

	Coloring any input graph properly with $\Delta + 1$ colors can be computed within $O(\Delta + \log{n})$ rounds in the one-bit model. Moreover, obtaining a $(\Delta + 1)$-coloring from $O(\Delta^2)$-coloring in this model requires $O(\Delta)$ rounds.

	\end{col}

The last assertion of Corollary \ref{col:graphd2col} follows from Corollary \ref{col:squarerootcoloring} and Lemma \ref{lem:enoughcolorsonebit}. 

\subsection{Computing $O(\Delta \cdot k)$ Coloring within $O(\Delta / k)$ Rounds} \label{sc:re}
In this section we describe a minor change in AG algorithm that applies to the CONGEST, LOCAL and SET-LOCAL models. (It will not apply to the one-bit model). This way, a faster computation is performed, in the expense of increasing the number of colors.  Specifically, for an integer $k$, such that $1 \leq k < \Delta$, we compute $O(\Delta \cdot k)$-coloring within $O(\Delta / k)$ rounds, starting from an $O(\Delta^2)$-coloring.
The change we suggest is to use triplets instead of ordered pairs for representing colors.

We provide the pseudocode of the algorithm below. (See Algorithm \ref{alg:refine}.) Next, we analyze the algorithm.
The algorithm starts with a proper $O(\Delta^ 2)$ coloring, where each color is represented by a triplet $\langle a_v,b_v,c_v \rangle$, such that $a_v,b_v \in \{0,1,....,q - 1\}, q = O(\Delta)$, $c_v = 0$, where $q$ is a prime. During an execution the colors change, but it always holds that $0 \leq a_v, b_v < q$ and $0 \leq c_v < k < \Delta < q$. The pseudocode describes the steps performed in a single round. The same steps are executed in  every round $i = 0,1,2,...$.

\begin{algorithm}
	\caption{Refine-AG} \label{alg:refine}
	\begin{algorithmic}[1]
		
		\STATE Let $\psi_i(v) = \langle a_v,b_v,c_v \rangle $ be the current color // $0 \leq c_v < k$, initially $c_v = 0$
		
		// Invariant: $a_v = 0$ or $c_v = 0$
		
		\IF {$a_v \neq 0$}
		\IF {exists an index $j, 0 \leq j < k$, such that the following two conditions hold:

			\STATE 1. for all neighbors $u$ of $v$ with $a_u \neq 0$:
			
			\STATE $ (b_v + j \cdot a_v) \bmod q  \neq  (b_u + j  \cdot a_u) \bmod q $
			
			\STATE and
			
			\STATE 2. for all neighbors $u$ of $v$ with $a_u = 0$:
			
			\STATE $\langle (b_v + j \cdot a_v) \bmod q, j \rangle \neq	\langle b_u , c_u \rangle $
			
		}	 		
		
		\STATE $\psi_{i + 1}(v) = \langle 0, (b_v + j \cdot a_v) \bmod q, j \rangle$
		
		\ELSE 
		
		\STATE $\psi_{i + 1}(v) = \langle a_v, (b_v + k  \cdot a_v) \bmod q, 0 \rangle$
		\ENDIF
		
		\ENDIF		
		
		\IF {$a_v = 0$}
		
		\STATE $\psi_{i + 1}(v) = \psi_{i}(v)$
		
		\ENDIF
		
		\STATE send $\psi_{i + 1}(v)$ to all neighbors of $v$
		
		\STATE receive the colors $\psi_{i + 1}$ of all neighbors of $v$
		
	\end{algorithmic}
\end{algorithm}

Next we argue that the algorithm maintains a proper coloring throughout its execution.

\begin{lem} \label{lem:refineproper}
	Given an arbitrary graph $G = (V, E)$ with a proper $O(\Delta ^ 2)$ coloring, Refine-AG produces a proper coloring after every round.
\end{lem}
\begin{proof}
	The proof is by induction on the number of rounds/iterations.\\ {\bf Base:} The initial coloring is proper.\\
	{\bf Step:}	We assume that in iteration $i$ the coloring is proper. Next, we show that it is also proper in iteration $i + 1$. Fix a vertex $v$. For a positive integer $x$, we denote the values of $a_v,b_v,c_v$ in iteration $x$ by $a_{v_x},b_{v_x},c_{v_x}$, respectively, i.e., $\psi_x(v) = \langle a_{v_x},b_{v_x},c_{v_x} \rangle$. If line 9 of the algorithm was executed then $\psi_{i + 1}(v) = \langle 0, (b_{v_i} + j \cdot a_{v_i}) \bmod q, j \rangle$, for some index $j$, $0 \leq j < k$. Since $a_{v_{i+1}} = 0$, $\psi_{i+1}(v)$ cannot be equal to the chosen colors in iteration $i + 1$ of any of $v$'s neighbors $u$ that executed line 11, simply because $a_{u_{i+1}} \neq 0$. Thus, assume that another neighbor $u$ executed line 9 and caused a conflict. This means that both nodes have the same index $j$, and
	$ (b_{v_{i}} + j \cdot a_{v_{i}}) \bmod q = (b_{u_{i}} + j \cdot a_{u_{i}}) \bmod q $. 
	But this is impossible, since the if statement in lines 3-5 prevents it.
	
	It is left to analyze the case that both neighbors execute line 11. This means that $a_{v_{i+1}} \neq 0$ and $a_{u_{i+1}} \neq 0$. Thus $u$ and $v$ have not executed line 9 before. The values of $c_u$ and $c_v$ can become non-zero only in line 9. Therefore, $c_{u_i} = c_{v_i} = 0$. Hence, if $\langle a_{v_i}, (b_{v_i} + k  \cdot a_{v_i}) \bmod q, 0 \rangle = \langle a_{u_i}, (b_{u_i} + k  \cdot a_{u_i}) \bmod q, 0 \rangle$, then $\langle a_{v_i}, b_{v_i}, c_{v_i} \rangle = \langle a_{v_i}, b_{v_i}, c_{v_i} \rangle$. This is a contradiction to the correctness of the coloring in round $i$.
\end{proof}

The next lemma helps us to show that only a bounded number of conflicts can occur throughout the execution of Algorithm \ref{alg:refine}.
\begin{lem} \label{lem:refine_conflict_iteration}
	For $(u,v) \in E$ with $a_u \neq 0, a_v \neq 0$, in each round there can be at most one index $j$,
	such that  $ (b_v + j \cdot a_v) \bmod q = (b_u + j \cdot a_u) \bmod q$. 
\end{lem}
\begin{proof}
Assume for contradiction that there are two indices $j_1 > j_2$, such that \\
(1) $(b_v + j_1 \cdot a_v) \bmod q = (b_u + j_1 \cdot a_u) \bmod q$\\ and \\
 (2) $(b_v + j_2 \cdot a_v) \bmod q = (b_u + j_2 \cdot a_u) \bmod q$.\\
By subtracting $(2)$ from $(1)$ we get  $(j_1 - j_2)(a_v - a_u) \equiv 0 ( \mbox{ mod } q)$. 
Since $0 < j_1 - j_2 < k < q$, it follows that $j_1 - j_2 \not\equiv 0$ $ (\mbox { mod } q)$, and thus $a_u \equiv a_v$ $ ( \mbox{ mod } q)$. \\
\noindent Then, from (1) it follows that $b_u \equiv b_v$ $( \mbox{ mod } q)$. But this means that $\langle a_u, b_u , c_u \rangle = \langle a_v, b_v , c_v\rangle$, since $a_u \neq 0, a_v \neq 0$ implies $ c_u = c_v = 0$. However, this is a contradiction to the correctness of the coloring in each round. 
\end{proof}
We say that a node $u$ {\em conflicts} with its neighbor $v$, if
$a_v \neq 0$, and there exist an index $j$, $0 \leq j \leq k$, such that either ($a_u \neq 0$ and $b_v + j \cdot a_v  \equiv b_u + j \cdot a_u (\bmod q)$) , or ($a_u = 0$ and $ \langle b_u + j \cdot a_v (\bmod q), j \rangle = \langle b_u, c_u \rangle$).
Next, we analyze how many times a node $u$ can conflict with a neighbor $v$ during an execution of $O(\Delta/k)$ rounds of Refine-AG.
\begin{lem} \label{lem:refine_conflict_algorithm}
	A node $u \in V$ can conflict with a neighbor $v$ of $u$ at most twice during $\left \lfloor q/k \right \rfloor - 1$ rounds of Refine-AG.
\end{lem}
\begin{proof} If $a_u = a_v \neq 0$, then $b_u \neq b_v$. Then for any index $j$, $0 \leq j \leq k$, we have $b_v + j \cdot a_v  \not\equiv b_u + j \cdot a_u (\bmod \ q)$,
 and no conflict occurs. 
Consider now the case that $a_v \neq 0$ and $a_u \neq 0$, and $a_v \neq a_u$. As long as both $v$ and $u$ keep being in non-final states (having their first coordinates  different from 0), we argue that once a conflict between them occurs, the  next conflict between them can happen only after at least $q/k - 1$ rounds.
Indeed, if $v$ and $u$ as above conflict at a certain round, it means that there exists an index $j$, $0 \leq j < k$, such that $b_v + j \cdot a_v  \equiv b_u + j \cdot a_u (\bmod \ q)$. On each of the subsequent rounds (as long as both $v$ and $u$ did not finalize their colors), we will have their respective second coordinate  $b_v$ and $b_u$ increase by $k \cdot a_v$ and by $k \cdot a_u$, respectively. As a result, if a conflict occurs again after some $h$ rounds, for an integer $h \geq 1$, then we have: $(b_v + k \cdot h \cdot a_v) + j' \cdot a_v  \equiv (b_u + k \cdot h \cdot a_u) + j' \cdot a_u) (\bmod \ q)$, for some $j', 0 \leq j' < k$.
Denote $z \equiv b_v + j\cdot a_v \equiv b_u + j\cdot a_u (\bmod \ q)$. we have
\begin{equation} \label{eq:pr}
z + (k \cdot h + j'-j) \cdot a_v \equiv  z + (k \cdot h + j'-j) \cdot a_u (\bmod \ q). 
\end{equation}
As $0 \leq j, j' \leq q$,
 for $h < q/k - 1$, we have $k \cdot h + (j'-j) < k (h +1) < q$. Thus, equality in equation (\ref{eq:pr}) can only happen if $a_v \equiv a_u (\bmod \ q)$. This is, however, a contradiction.
Thus, during rounds indexed $h$ with $0 \leq h \leq q / k -1$, at most one conflict can occur between non-finalized vertices  $v$ and $u$.

A conflict can also occur if $u$ is in non-final state and $v$ is in a final state.
i,e, $\langle b_u + j \cdot a_v, j \rangle = \langle b_v, c_v \rangle$, with $a_u \neq 0$ and $a_v = 0$. Observe that from that point on, the vertex $v$ will not change its color, and thus a conflict can occur only with the same index $j$.
On every round (as long as $u$ does not finalize), $k \cdot a_u$ is added to $b_u$.
Hence for both $b_u + j \cdot a_u$ and $b_u + k \cdot h \cdot a_u + j \cdot a_u$ to be equal to the same value $b_v$ (in $Z_q$), we must have $k \cdot h \geq q$. Hence within $q / k - 1$ rounds, at most one conflict of this kind can occur.
Thus, overall $v$ and $u$ may be in conflict at most twice, during $q/k - 1$ rounds.
\end{proof}

We now ready to summarize the properties of Algorithm \ref{alg:refine} (Procedure Refine-AG).
\begin{col}
	Refine-AG produces a proper coloring using $O(\Delta \cdot k)$ colors within $O(\Delta / k)$ rounds, starting from an $O(\Delta^2)$-coloring.
\end{col}
\begin{proof}
Consider the total number of pairs $(u, R)$, where $u$ is a neighbor of $v$ that conflicts with $v$ on round $R$. Denote this number by $N$. We have $N \leq 2 \cdot \Delta$, but also $N \geq k \cdot h$.
The former inequality is because every neighbor can belong to at most two such pairs, as long as the number of rounds on which the color of $v$ did not finalize satisfies $h \leq q/k -1$. The latter inequality assumes that in each of the $h$   rounds, the vertex $v$ had at least $k$ conflicts, and thus did not finalize. Thus $h \leq 2 \cdot \Delta / k$. In fact, we run the algorithm for one more round, i.e, for $\lceil 2 \cdot \Delta / k \rceil + 1$ rounds, to ensure that there will be a round in which there exists an index $j \in [0, k - 1]$ for which $v$ has no conflict.
We select $q$ to satisfy $\lceil 2 \cdot \Delta / k \rceil + 1 \leq q / k - 1$, i.e., $2 \cdot \Delta / k + 3 \leq q / k$. Hence we set $q \geq 2 \cdot \Delta + 3 \cdot k$. This guarantees that every vertex finalizes within $\lceil 2 \cdot \Delta / k \rceil + 1$ rounds.

\end{proof}

To implement this algorithm using bit-messages we can send on every round a single bit indicating if the vertex $v$ (that runs the algorithm) finalizes or not, and if it does finalize, we append the value of $j$ with which $v$ finalizes to the message. Overall, the algorithm requires every vertex to send $O(\Delta/k)$ bit-messages and one single message of size $O(\log k)$. Thus, the algorithm can be implemented in $O(\Delta/ k \cdot \log k)$ bit rounds.

\section{Fully-Dynamic Self-Stabilizing algorithms with $O(\Delta + \log^*n)$ rounds}
{\noindent \large \bf 4.1 Fully-Dynamic Self-Stabilizing $(\Delta + 1)$-Coloring\\}
In this section we employ a variant of Linial's algorithm for $O(\Delta^2)$-coloring that allows a vertex $v$ to avoid being colored by colors from a given set $R(v)$ of size at most $O(\Delta)$ 
\cite{B15}. (This is useful when selecting a new color, to avoid collisions with some neighbors that have already obtained final colors.) We refer to this algorithm as Algorithm Excl-Linial. 

	Algorithm Excl-Linial is identical to Linial's original algorithm, except for the final stage that transforms a proper $O(\Delta^3)$-coloring into a proper $O(\Delta^2)$-coloring. In this stage each vertex $v$ computes a polynomial $P_v(x)$ of degree $2$ in a field of size $O(\Delta)$, and selects a color $\langle x,P_v(x) \rangle$, such that $\langle x,P_v(x) \rangle \neq \langle y,P_u(y) \rangle $, for any neighbor $u$ of $v$ and any $y$ in that field. Since the degree of the polynomials in this stage is $2$, each polynomial intersects with a neighboring node's polynomial in at most two points. Hence, there are at most $2\Delta$ points on $P_v$ that may intersect with some neighbor. If the field is of size at least $2\Delta + 1$, there must be a point such that $\langle x,P_v(x) \rangle \neq \langle y,P_u(y) \rangle$ for all neighbors $u$ of $v$ and all elements $y$ in the field. Such a pair is selected by the original algorithm of Linial. In the modified variant, on the other hand, the field is of size greater than $3\Delta$.
	Consequently, if a set $R(v)$ of at most $\Delta$ forbidden colors is provided, there still exists an element $x$ in that field, such that $\langle x,P_v(x) \rangle$ is not equal to any of the colors in the set $R(v)$, and neither to any $\langle y,P_u(y) \rangle$, for a neighbor $u$ and an element $y$. Such a color is selected as a final color. Thus, we obtain an $O(\Delta^2)$-coloring, where all colors belong to sets that exclude $O(\Delta)$ colors each, within $\log^* n + O(1)$ time.
More generally, if a forbidden set of colors is of size $c \cdot \Delta$, for some constant $c  > 0$,  then Algorithm Excl-Linial works in the same way, but uses a field of size at least $(c + 2) \cdot \Delta$.	
	This completes the description of algorithm Excl-Linial. 
\\
\indent Before describing our self-stabilizing algorithm, we define some notation, and describe yet another useful variant of Linial's algorithm, which we call Algorithm Mod-Linial.
Let $r = \log^* n + O(1)$ denote the number of iterations in Linial's algorithm. Let $t_r = O((\Delta \log{n})^2), t_{r - 1} = O((\Delta (\log \Delta + \log \log{n})^2), ..., t_1 = O(\Delta^2) $ denote upper bounds on the number of colors in the different iterations of Linial's algorithm. 
Define the intervals $I_0,I_1,I_2,...$ as follows. $I_0 = [0, t_1 - 1], I_1 = [t_1 ,t_1 + t_2 - 1],...,
I_{r - 1} =  [ \ \sum_{i = 1}^{r-1} t_i, \ \sum_{i = 1}^r t_i - 1], $ 
$I_r =  [ \ \sum_{i = 1}^{r} t_i, \ \sum_{i = 1}^r t_i + n - 1].$
Since each such interval contains a sufficient number of colors, we can map each color palette of each iteration of Linial's algorithm to one of the intervals defined above. Specifically, the palette of the first iteration is mapped to $I_{r - 1}$ (which is of size $t_r$), the palette of the second iteration is mapped to $I_{r -2}$ (which is of size $t_{r - 1}$), and so on, up to the last palette that is mapped to $I_0$. This way Linial's algorithm is modified, so that in each iteration $i = 1,2,...,r$ a coloring using a palette $I_{r -i + 1}$ is transformed into a coloring using the palette  $I_{r  -i}$. (The actual number of colors used from this palette is $O((\Delta \log^{(i)} n)^2)$.) The modified algorithm will be referred to as Mod-Linial.
It accepts as input a color of a vertex $v$, a (sub)set of its neighbors colors, and a set of $O(\Delta)$ forbidden colors, and returns a new color for $v$.
The range $I_r =  [ \ \sum_{i = 1}^{r} t_i, \ \sum_{i = 1}^r t_i + n - 1]$ will be used for an initial $n$-coloring obtained from IDs.

Observe that the idea described above in algorithm Excel-Linial, can be easily incorporated into algorithm Mod-Linial as well. Specifically, on each iteration $i = 1,2,...,r$ of algorithm Mod-Linial, every vertex $v$ evaluates a polynomial $P_v$. Consequently, two polynomials $P_v$ and $P_u$ of neighboring vertices $v$ and $u$ may agree in at most a certain pre-determined number of values. The polynomials are over the field $GF(q)$, for $q = O(\Delta)$ being a prime (characteristic of $GF(q)$). By increasing this characteristic by an additive $c \cdot \Delta$ term, for a constant $c > 0$, one can ensure that the chosen color for $v$ will exclude a list $Q_v$ of at most $c \cdot \Delta$ forbidden colors.

Our fully-dynamic self-stabilizing algorithm works as follows. The RAM of each vertex consists of a variable that holds a color in a range $\{0,1,...,t_1 + t_2 + ...  + t_r  + n -1 \}$.  The ROM of each vertex holds the algorithm, the number of vertices $n$ and the maximum degree $\Delta$. In each round each vertex $v$ checks whether it is in a proper state, i.e., its color is different from colors of all its neighbors. (See the pseudocode of Procedure Check-Error below.) If $v$ is not in a proper state, the vertex returns to its initial state. (See lines 4- 5 of Procedure Self-Stabilizing-Coloring.) We define the initial state of a vertex with ID $j \in {0,1,...,n-1}$ by the color  $t_1 + t_2 + ...  + t_r  + j$. Otherwise (i.e., if Procedure Check-Error returned that its color is different from colors of all its neighbors),  the vertex is in a proper state. Then, the vertex $v$ computes its next color or finalizes the current one. (See lines 7 - 23 of Procedure Self-Stabilizing-Coloring.) Specifically, as long as the vertex color belongs to an interval $I_j$ for $j \geq 2$, i.e., the color is significantly larger than $\Delta^2$, the vertex computes the next color from a smaller range using the algorithm Mod-Linial (lines 9-10 of Procedure Self-Stabilizing-Coloring). Once a color is in the interval $I_1$, the vertex must select a new color in the interval $I_0$ that is distinct from any neighboring color that is also in $I_0$. This is done in lines 12 - 14 of the procedure.
The set $S'$, computed in line 13 and provided as the third parameter of Procedure Mod-Linial in line 14, contains all possible colors that neighbors $u$ of $v$ that run already lines 15 - 21 (i.e., their colors are small enough) may obtain in the current iteration. Note that for each such $u \in \Gamma(v)$ there are at most $2$ such colors.
Finally, a color that is in the range $I_0$ either becomes final or changes to another color in $I_0$ according to Algorithm AG. See lines 15 - 21. This completes the description of the algorithm. Its pseudocode is provided below. Next, we analyze the algorithm.

\begin{algorithm}
	\caption{Check-Error ($my\_color$, $[neighbors\_colors]$)}
	
	\begin{algorithmic}[1]
		\IF {$ my\_color \in neighbors\_colors$} 
		
		\STATE return error
		
		\ENDIF
		
		\STATE return valid
	\end{algorithmic}
\end{algorithm}

\begin{algorithm}
	\caption{Self-Stabilizing-Coloring (run by every vertex $v$ in parallel)} \label{alg:selfstcol}
	\begin{algorithmic}[1]
	
    \STATE clear buffers of incoming and outgoing messages
			
		\STATE send $my\_color$, $my\_ID$ to all neighbors
		
		\STATE receive the colors and IDs of all neighbors, and store colors in $[neighbors\_colors]$, such that any color of a neighbor $u$ that is greater than $t_1 + t_2 + ...  + t_r$ is replaced with $t_1 + t_2 + ...  + t_r + ID_u$
		
		\IF {Check-Error($my\_color$, $[neighbors\_colors]$) = error \OR $my\_color > t_1 + t_2 + ...  + t_r$ }
		
		\STATE $my\_color = t_1 + t_2 + ...  + t_r + my\_ID$ \ \ \ /* initial state */ 
		
		\ELSE
		
		\STATE Let $I_j$ denote the range that $my\_color$ belongs to
		
		\STATE Let $Q$ denote the subset of $[neighbors\_colors]$ of all colors that belong to $I_j$
		
		\IF {$j \geq 2$}
		
		\STATE $my\_color$ = Mod-Linial($my\_color, Q, \emptyset$)
		
		\ELSIF {$j = 1$}
		
		\STATE Let $S$ denote the subset of $[neighbors\_colors]$ of all colors that belong to $I_0$, represented as ordered pairs
		
		\STATE Let $S' = \{ \langle a, (b + a) \mbox{ mod } q\rangle  \ | \ \langle a,b\rangle  \in S \} \cup  \{ \langle 0, b\rangle  \ | \ \langle a,b\rangle  \in S \}$
		
		\STATE $my\_color$ = Mod-Linial($my\_color, Q, S'$) /* avoid collisions with $S'$ */
		
		\ELSIF {$j = 0 $}
		
		\STATE represent $my\_color$ as an ordered pair $\langle a, b\rangle $
		
		\IF {$\langle a,b\rangle $ conflicts with a color in $Q$}
		
		\STATE $my\_color = \langle a, (a + b) \mbox{ mod } q\rangle $
		
		\ELSE
		
		\STATE $my\_color = \langle 0, b\rangle $  /* final color */
		
		\ENDIF
		
		\ENDIF
		
		\ENDIF
		
	\end{algorithmic}
\end{algorithm}

\clearpage

We start with the observation that the submodules invoked by the algorithm are self-stabilizing.
\begin{lem} \label{lem:modulstab}
The procedures Check-Error and Mod-Linial are self-stabilizing.
\end{lem}
\begin{proof}
Procedure Check-Error is applied solely with RAM variables, without message exchange. Consequently, in the period when faults no longer occur, the procedure returns 'valid' iff the value of the variable $my\_color$ does not appear in any entry of the collection $[neighbors\_colors]$. In other words, this procedure performs the desired operation, regardless of the actions of the adversary during the period of faults.

Procedure Mod-Linial is applied solely with RAM variables as well. Specifically, it is invoked with the variable $my\_color$ and two collections $Q,S'$ of variables. There are two possibilities: either the preconditions of the procedure apply, and then it returns a proper solution, or they do not apply, and it returns some value, which may be wrong. However, once faults no longer occur, the execution of procedure Check-Error before procedure Mod-Linial guarantees that $my\_color$ does not belong to $Q$. This, in turn, guarantees that $I_j$ is computed properly in line 7. This results in a correct construction of the sets $Q$, $S$ and $S'$ in lines 8,12,13.  Therefore, the preconditions of Procedure Mod-Linial hold, and as will be shown in the sequel, it returns a proper value.
\end{proof}

Observe that the set $S'$, computed in line 13, is the set of all possible colors that neighbors of the vertex $v$ (that runs the algorithm) that already executed algorithm AG (in lines 15 -21) may obtain on the next round. The set $S$ is the set of the current colors of these neighbors.


We start with arguing that the algorithm maintains a proper coloring.

\begin{lem} \label{lem:selfstabilizinglinialcorrect}
	Given an arbitrary graph $G = (V, E)$, our self-stabilizing algorithm produces a proper coloring $\psi(G)$ in each round, once faults no longer occur.
\end{lem}
	\begin{proof}
		Consider a round $i$. If a node $v \in V$ has a color that is equal to that of a neighbor $u$, i.e., $\psi_i(u) = \psi_i(v)$, then (by line 5 of Algorithm \ref{alg:selfstcol}), $\psi_{i+1}(v) = t_r + t_{r - 1} + ...  + t_1 + id(v) \ne \psi_{i+1}(u) = t_r + t_{r - 1} + ...  + t_1 + id(u)$. In this case $\psi_{i+1}(v)$ must be different from the colors $\psi_{i+1}$ of all neighbors $u$ of $v$, since their colors either become at most $t_r + t_{r - 1} + ...  + t_1$ or become equal to  $t_r + t_{r - 1} + ...  + t_1 + id(u) \neq \psi_{i+1}(v)$. 
		
		Otherwise, lines 6 - 23 are executed.
		Since it is assumed that no more faults will occur, we prove that lines 6-23 provide a proper coloring. If $j \geq 2$ (line 10) then $\psi_{i+1}(v)$ will be in the range $I_{j-1}$. (Any element in $I_j$ is greater than any element in $I_{j-1}$, and thus numerical values of colors decrease as the algorithm proceeds. Also, note that all intervals are disjoint.) Therefore, all neighbors $u$ with $\psi_i(u) \not \in I_j$ will not select a new color $\psi_{i + 1}(u)$ from $I_{j - 1} $. For a neighbor $u$ with $\psi_i(u) \in I_j$, its color belongs to $Q$, and  Mod-Linial algorithm will produce a proper coloring.\\
		If $j = 1$ then Procedure Mod-Linial works in the following way. It computes a new color from $t_0$, such that it is distinct from all neighbors' colors that transit from $I_1$ to $I_0$ in round $i$, and from all colors of the set $S'$. The latter set contains all possible colors that can be used in round $i + 1$ by neighbors of $v$ with colors in the range $I_0$ in round $i$. Consequently, the new color of $\psi_{i + 1}(v)$ of $v$ is distinct from the new colors of such neighbors. Moreover, the new color is also distinct from new colors of the rest of the neighbors, since they were either in ${I_1}$ in round $i$, and do not collide with $v$ in round $i + 1$ due to correctness of Mod-Linial, or in a higher range, and thus are not in ${I_0}$ in round $i + 1$.\\ 
		If $j = 0$, then lines 15 - 22 execute our Additive-Group algorithm (see Lemma \ref{lem:propercol} and Corollary \ref{col:squarerootcoloring}), and produce a proper coloring for neighbors with $j = 0$. For neighbors with $j > 0$, the coloring is proper as well, by analysis of previous cases in this proof.
	\end{proof}

Next we analyze the quiescence (i.e., stabilization) time of our algorithm.
\begin{lem} \label{lem:selfstabilizinglinialfinish}
	Given an arbitrary graph $G = (V, E)$, our fully-dynamic self-stabilizing algorithm produces a proper $O(\Delta)$-coloring with $O(\Delta + \log^*n)$ stabilization time. 
\end{lem}
	\begin{proof}
		By induction on $i$, it is easy to see that in the end of each round $i = 1,2,...$, counting from the moment that faults stop occurring, all colors are in the range $I_0 \cup I_1 \cup ... \cup I_{r + 1 -i}$. Therefore, within $r + 1 = \log^* n + O(1)$ rounds, all colors are in the range $I_0$, and the coloring is proper. From this moment and on, the procedure executes our Additive-Group algorithm in all vertices. Therefore, by Corollary \ref{col:squarerootcoloring}, within $O(\Delta)$ additional rounds the number of colors becomes $O(\Delta)$.
\end{proof}

We also obtain a self-stabilizing algorithm that employs exactly $(\Delta + 1)$ colors. To this end, in each round each vertex $v$ with a color of the from $\langle 0,b_v \rangle, b_v > \Delta,$ whose all neighbors also have 0 in the first coordinate of their colors performs the following. If $\langle 0,b_v \rangle$ is greater then the colors of all $v$'s neighbors, then $v$ selects a new color $\langle 0, b_v' \rangle$ such that $0 \leq b'_v \leq \Delta$, and $\langle 0,b'_v \rangle $ is distinct from all colors of $v$'s neighbors. Consequently, once all colors in the graph are of the form $\langle 0,b \rangle$, $b = O(\Delta)$, at most $O(\Delta)$ additional rounds are required to arrive to a $(\Delta + 1)$-coloring, because at least one color is eliminated in each round. (This is the greatest color, as long as there are colors greater than $\Delta$.) Moreover, starting from any configuration of the RAM values, in any round the produced coloring is proper, and the color ranges decrease as in the $O(\Delta)$-coloring algorithm. Thus, within $O(\Delta + \log^*n)$ rounds all vertices enter the range of colors of $O(\Delta)$, and within additional $O(\Delta)$ rounds we obtain a $(\Delta + 1)$-coloring.
Alternatively, the same effect can be achieved via our 1-bit halving reduction, described in Section \ref{sc:smbit}.
We summarize this below.

\begin{thm} \label{stcol}
	Given an arbitrary graph $G = (V, E)$, our fully-dynamic self-stabilizing algorithm produces a proper $(\Delta + 1)$-coloring with $O(\Delta + \log^*n)$ stabilization time.
\end{thm}

Note that the proof above applies in a fully dynamic setting. Specifically, the edges may appear and fall, vertices can connect and disconnect, but as long as upper bounds on $n$ and $\Delta$ are hard-wired in the ROM and are not violated, the algorithm will stabilize to a $(\Delta + 1)$-coloring. (Though, admittedly, this $\Delta$ will be just an upper bound on the current maximum degree of the graph, which can obviously be much smaller.)

{\noindent \large \bf 4.2 Fully-Dynamic Self-Stabilizing MIS, MM, and ($2\Delta - 1)$-Edge-Coloring\\}
We employ our self-stabilizing coloring algorithm from the previous section in order to compute MIS as follows. We add a bit $\mu_v$ to the RAM of each vertex $v \in V$. This bit represents whether $v$ is in the MIS (if $\mu_v = 1$)  or not in the MIS (if $\mu_v = 0$). We add the following instruction in the end of Procedure Self-Stabilizing-Coloring. If all neighbors $u$ of $v$ with smaller colors than that of $v$ have $\mu_u = 0$, then we set $\mu_v = 1$. Otherwise, we set $\mu_v = 0$. This completes the description of the changes required to compute an MIS. Denote by $U$ the vertex set computed by this algorithm. 

The next theorem shows that within $i$ rounds, for $i > 0$, after the stabilization of coloring, all vertices with colors $1,2,...,i$ induce a subgraph with a properly computed MIS. Consequently, within $O(\Delta)$ additional rounds an MIS of the entire input graph is constructed.

\begin{thm} \label{stmis}
	Given an arbitrary graph $G = (V, E)$, our self-stabilizing algorithm produces a proper MIS within $O(\Delta + \log^*n)$ rounds after the last fault.
\end{thm}

\begin{proof}
	Let $t_{cd} = O(\Delta + \log^* n)$ be the stabilization time of the coloring algorithm. (See Theorem \ref{stcol}.) Denote by $U_i$, $i = 1,2,...,\Delta + 1$, the set of vertices $v$ that belong to MIS (i.e., have $\mu_v = 1$) at round $t_{cd} + i$ after faults stop occurring. Let $\psi$ be the $(\Delta + 1)$-coloring maintained by the algorithm. (We know that $t_{cd}$ rounds after the last fault occurred, $\psi$ is indeed a proper $(\Delta + 1)$-coloring.)
	
	We prove by induction on $i$ that at time $t_{cd} + i$ after faults stop occurring, for $i = 1,2,....,\Delta + 1$, $U_i$ is an MIS for the set $\hat{V}_i = \{ v \ | \  1 \leq \psi_i(v) \leq i \}$, where $\psi_i$ is the coloring $\psi$ maintained by the algorithm at that time.\\
	{\bf Base (i = 1):} All vertices of $\hat{V}_1$ form an independent set (because $\varphi_1$ is a proper $(\Delta + 1)$-coloring, because it is the coloring $\psi$ more than $t_{cd}$ rounds after the last fault occurred, and each of them joins MIS because they have no neighbors of smaller color).\\
	{\bf Step:} For some $i \leq \Delta$ we assume that $U_i$ is an MIS for $\hat{V}_i$. Consider a vertex $v \in U_{i + 1}$, i.e., $\psi_{i + 1}(v) = i + 1$. This vertex had the same color $i + 1$ for all the rounds $t_{cd} + 1, t_{cd} + 2,....,t_{cd} + i + 1$, counting from the moment $T$ when faults stopped occurring. By end of round $T + t_{cd} + i$ or earlier, all its neighbors of smaller color (they also did not change their colors during the time interval $[T + t_{cd}, T + t_{cd} + i]$) have set their values $\mu_u$. So in round $T + t_{cd} + i + 1$, if $v$ has no neighbor with a smaller color in  the MIS, it joins MIS. (It might have joined earlier, but it will anyway check again whether it has to join in round $T + t_{cd} + i + 1$.) Since vertices of $V_{i + 1} = \{ v \ | \ \psi_{i + 1}(v) = i + 1 \}$ form an independent set, the resulting set $U_{i + 1}$ is a maximal independent set for $\hat{V}_i \cup V_{i + 1} = \hat{V}_{i + 1}$.
\end{proof}

In the ordinary (non-stabilizing) setting it is possible to compute a maximal matching and an edge coloring by simulating the line-graph of the input graph, and computing an MIS and vertex-coloring of it. These solutions on the line graph directly provide solutions for maximal matching and edge coloring of the input graph within the same running time. This technique is applicable also to the self-stabilizing setting. Specifically, each vertex $v$ simulates virtual vertices, one virtual vertex per edge adjacent on $v$. In the beginning of each round each vertex verifies whether the state of each of its virtual vertices that correspond to some edge equals to the state in the other endpoint of that edge. If this is not the case, the endpoint with a greater ID copies the state of the other endpoint for that virtual vertex. Consequently all edges have consistent representations, i.e., the same state in both their endpoints, in the entire graph. Now, a self-stabilizing MIS or vertex-coloring algorithm can be simulated correctly on the line graph in order to produce self-stabilizing maximal matching and edge-coloring of the input graph. In conjunction with Theorems \ref{stcol}, \ref{stmis} this leads to the following result.
\begin{thm}
	Given an arbitrary graph $G = (V, E)$, our self-stabilizing algorithms produce a maximal matching and a proper $(2\Delta - 1)$-edge-coloring within $O(\Delta + \log^*n)$ stabilization time.
\end{thm}
We remark, however, that while our self-stabilizing vertex-coloring and MIS algorithms require small messages, this is not the case for the edge-coloring and maximal matching algorithms.

\section{Edge-Coloring}
\subsection{Edge Coloring within $O(\Delta + \log^*n)$ Rounds in the CONGEST Model and $O(\Delta + \log n)$ Rounds in the Bit-Round Model}
Next, we employ our techniques in order to compute {\em edge colorings} using small messages. The algorithm consists of two stages. The first stage constructs an $O(\Delta^2)$-edge-coloring from scratch, and the second stage computes an $O(\Delta)$-coloring from this $O(\Delta^2)$-coloring. We remark that we cannot use the algorithm of Linial for the first stage, since its message complexity in the case of edge-coloring is quite large. Instead, we do the following. We invoke Kuhn's algorithm \cite{K09} for $2$-defective $\Delta^2$-edge coloring. 
This algorithm orients all edges towards endpoints with greater IDs. Then, each vertex assigns its outgoing edges distinct colors from the set $\{1,2,...,\Delta\}$. It also assigns its incoming edges distinct colors from the same range. Consequently, each edge obtains a pair of colors, one color from each of its endpoints.
This is done within a single round by sending a message of size $O(\log n)$ per edge (in both directions).
These messages contain vertex IDs.

Each color of an edge $e \in E$ can be represented as an ordered pair $\psi(e) = \langle i, j\rangle $, where $i, j \in \{1,2,...,\Delta\}$. Note that a set of edges with the same $\psi$-color consists of paths and cycles, since each vertex on such an edge has at most one another edge adjacent on it in this set.
This is because the defect of $\psi$ is $2$. To remove the defect we run Cole and Vishkin coloring algorithm \cite{CV86} on edges of each color class in parallel and assign a new color to each $e \in E$ in the form $\psi(e) = (i, j, k)$. The first two indices $i,j$ are the result of the first stage, and the rightmost index $k \in \{1,2,3\}$ is the result of Cole-Vishkin's algorithm invocation.

Next, we compute an $O(\Delta)$-edge-coloring from the $O(\Delta^2)$-edge-coloring as follows. In each round both endpoints of an edge hold its color, that will be from now on represented as an ordered pair $\langle a,b\rangle $, $a,b \in O(\Delta)$, rather than a triple. Consequently, each endpoint can check for conflicts of edges adjacent on it. For each edge with a conflict at an endpoint, the endpoint that detects the conflict sends a message over this edge (consisting of a single bit) to notify the other endpoint about the conflict. Then, for each edge, both of its endpoints know whether it has a conflict with some adjacent edge or not. If the current edge color is $\langle a,b\rangle $, and there is a conflict, the new color becomes $\langle a, (a + b) \bmod q\rangle $. Otherwise, it becomes $\langle 0,b\rangle $. Both endpoints update the new color of their edge. This is done within a single round and by exchanging just a single bit on each edge. Then all vertices of the graph are ready to proceed to the next round and perform it in a similar way. The algorithm stops once all edges have colors of the from $\langle 0,b\rangle $, $0 \leq b < q = O(\Delta)$. (Here $q$ is a prime number that satisfies that the original number of colors is at most $q^2$ and $q \geq 2\Delta - 1$.)

\begin{lem} \label{edgecolor}
	A proper $O(\Delta)$-edge coloring is obtained in $O(\Delta + \log^* n)$ rounds in the CONGEST model.
\end{lem}
	\begin{proof}
		The algorithm starts with the invocation of Kuhn's algorithm that results in a $2$-defective $\Delta^2$-edge-coloring within $O(1)$ time. Then it is turned into a proper coloring using Cole-Vishkin algorithm within $O(\log^* n)$ time. Indeed, if prior to the execution of the latter algorithm a pair of adjacent edges had the same color $\langle i,j\rangle $, they now have distinct colors $\langle i,j,k\rangle $ and $\langle i,j,k'\rangle $, since Cole-Vishkin algorithm produces a proper $3$-coloring of the edges in the set of color class $\langle i,j\rangle $.
		Next, in each round each color of an edge of the form $\langle a,b\rangle $ is transformed either into $\langle a, (a + b) \bmod q\rangle $ or into $\langle 0,b\rangle $. In both cases the new coloring is proper. See Lemma \ref{col:squarerootcoloring}. Within $O(\Delta)$ rounds all colors obtain the form $\langle 0,b\rangle $.
	\end{proof}

In the next lemma we argue that the bit-complexity of our algorithm is small. The variant of CONGEST model in which vertices initially know the IDs of their neighbors is called $KT_1$ model. The variant in which they only know their own IDs is called $KT_0$ model \cite{KKT15}.  

\begin{lem} \label{edgebitcom}
	The bit complexity of our edge-coloring algorithm is $O(\Delta + \log n)$ per edge (in the $KT_0$ model). In addition, if initially vertices know the IDs of their neighbors (i.e., in the $KT_1$ model), then the bit complexity is $O(\Delta + \log \log n)$ per edge.
\end{lem}
	\begin{proof}
		Exchanging initial IDs between neighbors requires $O(\log n)$ bits. Exchanging the colors during the $2$-defective $\Delta^2$-edge-coloring requires $O(\log \Delta)$ bits. The first round of Cole-Vishkin algorithm is performed based on IDs of $O(\log n)$ bits. The second round of Cole-Vishkin algorithm requires $O(\log \log n)$ bits, the third one requires $O(\log \log \log n)$ bits, and so on. The last round of Cole-Vishkin algorithm requires a constant number of bits. The exchange between neighbors of the resulting proper $O(\Delta^2)$-edge coloring of the input graph requires $O(\log \Delta)$ bits. Each of the following $O(\Delta)$ rounds requires $1$ bit per message.  
	\end{proof}

We can also produce edge-coloring with exactly $(2\Delta - 1)$-colors as follows. Once the stage of $O(\Delta)$-edge-coloring terminates, we apply a procedure similar to One-bit AG halving reduction. (See Section \ref{sc:smbit}.) Specifically, let $k$ be the current number of colors, and $q = \left \lceil k/2 \right \rceil$. (Recall that in this algorithm $q$ does not have to be prime.) We represent each color of an edge as an ordered pair $\langle a_e, b_e \rangle$, where $a_e \in \{0,1\}$, $b_e \in \{0,1,..., q - 1 \}$. Then we execute $2\Delta$ rounds to halve the number of colors. In each round, for each edge $e = (u,v) \in E$, its endpoints $u,v$ check whether $b_e$ is distinct from all $b_{e'}$ of edges $e'$ adjacent on these endpoints. Then $v$ notifies $u$ whether this is the case for all edges adjacent on $v$. In parallel, $u$ notifies $v$ whether this is the case for all edges adjacent on $u$. If both $u$ and $v$ pass the check, they update the color of $e$ to $\langle 0, b_e \rangle$. Otherwise, they update it to $\langle 1, b_e + 1 \mbox{ mod } q \rangle$. Since each edge has at most $2\Delta - 2$ edges adjacent on it, within $2\Delta - 1$ rounds all edges $e \in E$ select a color with $a_e = 0$. (See Lemma \ref{lem:enoughcolorsonebit}.) Hence the number of colors is halved. Repeating this for a constant number of phases converts the $O(\Delta)$-edge-coloring into a $(2\Delta - 1)$-edge-coloring.
We summarize this below.

\begin{thm}
	We compute $(2\Delta - 1)$-edge-coloring within $O(\Delta + \log^*n)$ rounds in the CONGEST model, within $O(\Delta + \log \log n)$ rounds in the Bit-Round model with knowledge of neighbors' IDs ($KT_1$ model), and within $O(\Delta + \log n)$ time in the Bit-Round model without knowledge of neighbors' IDs ($KT_0$ model).
\end{thm}

\subsection{$(2  \Delta -1)$-Edge Coloring of Oriented Forests in the CONGEST Model}
In this section we devise a $(2\Delta - 1)$-edge-coloring algorithm for oriented forests that requires $\log^* n + O(1)$ rounds using only small messages, i.e., it can be executed in the CONGEST model.
We note that the currently existing algorithm for this problem, due to \cite{PR01}, requires messages of size $O(\Delta)$.
Our algorithm colors an input tree $T$ as follows. (The same algorithm applies to oriented forests as well.)

The algorithm starts with computing a $3$-vertex-coloring of $T$, via Cole-Vishkin algorithm, in ${\log^{*}n + O(1)}$ rounds. Denote the resulting coloring by $\varphi$.
Then we perform a shift-down (for just one round), to ensure that all siblings have the same color. (In the shift-down operation, all vertices $v \in V$, except the root $r$, adopt the color of their parent $\pi(v)$. The root $r$ selects 
a color from the set $\{1,2\}$, different from its current color.)
Then for all vertices $v \in V$, such that $\varphi(v) = 1$, run in parallel: color their descending edges by colors $1,2,..,\Delta -1$, except for the root $r$, that might have $\Delta$ descending edges. The root colors these edges with colors $1,2,..,\Delta -1$, and $2 \Delta - 1$  (if indeed $\deg(r) = \Delta$ and ${\varphi(r) = 1}$).\\
Next, for all vertices $v \in V$, such that $\varphi(v) = 2$, run in parallel: color their descending edges by colors $\Delta, \Delta + 1, \Delta + 2,..., 2 (\Delta - 1)$. (The root may need an additional color, which is $2\Delta - 1$.)\\
Finally, for all vertices $v \in V$ such that $\varphi(v) = 3$, color their descending edges as follows (after all descending edges of vertices with $\varphi(v) = 1$ and $\varphi(v) = 2$ have been already colored).\\
As $\varphi(r) \in \{1,2\}$, we have $v \neq r$. It means that $v$'s parent $\pi(v))$ (of $\varphi$-color 1 or 2) already assigned the edge $(\pi(v), v)$ a color. (Generally, when a vertex $v$ colors its descending edges, it informs the opposite endpoint of the color that the edge between them received.) The vertex $v$  also knows the $\varphi$-color of all its children. So, it knows that edges descending from its children are all colored by colors either from the set $\{1,2,...,\Delta - 1\}$, or from the set $\{\Delta, \Delta + 1,..., 2( \Delta -1)\}$. In either case, there are at most $\Delta$ forbidden colors from edges descending from $v$. In other words, there are at least $\Delta -1$ permitted colors. Also, there are at most $ \Delta -1$ edges descending from $v$. So, $v$ edge-colors them by these colors.


\begin{thm}
Our algorithm computes $(2\Delta - 1)$-edge-coloring of oriented $n$-vertex forests in $\log^* n + O(1)$ time, in the CONGEST model.
\end{thm}

\section{Arbdefective $O(\frac{\Delta}{p})$-coloring with defect $O(p)$}
Lovasz \cite{L66} showed that in a graph with maximum degree $\Delta$, there exists a $p$-defective $\frac{\Delta}{p}$-coloring, where $1 \le p \le \Delta$.
In this section we devise an algorithm for $O( \sqrt{\Delta})$-{\em arbdefective} $O(\sqrt{\Delta})$-coloring within $O(\sqrt{\Delta} + \log^* n)$ rounds.
More generally, our algorithm computes an $O(p)$-arbdefective $O(\Delta/p)$-coloring within time $O(\Delta/p + \log^* n)$. (Definitions of defective- and arbdefective-colorings can be found in Section 2.)
Our algorithm starts with computing an $O(\sqrt{\Delta})$-defective $O(\Delta)$-coloring. 
This is done using the algorithm of \cite{BEK14} within $O(\log^*n)$ rounds. 
(More generally, the algorithm of \cite{BEK14} computes a $p$-defective $O((\Delta/p)^2)$-coloring, for any positive parameter $p$, in $\log^* n + O(1)$ time.)
Then we perform $O(\Delta/p) = O(\sqrt{\Delta})$ rounds of color updates, rather than $O(\Delta)$ as in our Additive-Group algorithm. The update rule for arbdefective coloring is different from the rule for proper coloring. Specifically, we tolerate up to $p$ conflicts.
In other words, instead of setting $\psi_{i+1}(v) =  \langle 0, b\rangle  $ only if there are no neighbors with the same value $b$ in the second coordinate, we set this if there are at most $p = \Theta( \sqrt{\Delta})$ neighbors of different $\psi_i$-color with the same second coordinate $b$. We will show in the sequel that after $O(\Delta/p) = O(\sqrt{\Delta})$ rounds all colors are of the form $\langle 0,b\rangle $, and each color class induces a subgraph of arboricity $O(p)$. Thus, as a result we have an $O(\sqrt \Delta)$-arbdefective $O(\sqrt \Delta)$-coloring, and, more generally, an $O(p)$-arbdefective $O(\Delta/ p )$-coloring. The operations are performed in a field of a prime characteristic $q$, $q \geq 2 \lceil \Delta / p \rceil + 1$.  The pseudocode of the algorithm is provided below. The next lemmas analyze its running time and show its correctness. 

\begin{algorithm}
	\caption{Arbdefective-Color($G,v, p = \sqrt{\Delta}$)}
	\begin{algorithmic}[1]
		\STATE $\psi = $ compute an $O(p)$-defective $O((\Delta/p)^2)$-coloring of $G$ using \cite{BEK14}  /* $O(\sqrt \Delta)$-defective $O(\Delta)$-coloring */
		
		\STATE represent $\psi_0(v)$ as an ordered pair $\langle a,b\rangle $, such that $a,b \in O(\Delta / p)$. /*  $a,b \in O(\sqrt \Delta)$ */
		
		
		\STATE let $q = \Theta(\Delta/p)$ be the smallest prime such that $q$ is greater than $2\left \lceil \Delta/p \right \rceil + 1$
		
		\FOR {$i = 0,1,..., 2\left \lceil \Delta/p \right \rceil$}
		
		\IF {$v$ has at most $p$ neighbors $u$ of a different $\psi_i$-color, such that the second coordinate of $\psi_i(u)$ equals the second coordinate of $\psi_i(v)$}
		
		\STATE $\psi_{i+1}(v) = \langle 0,b\rangle $
		
		\ELSE
		
		\STATE $\psi_{i+1}(v) = \langle a, (a + b) \bmod q\rangle $
		
		\ENDIF
		
		\STATE send $\psi_{i+1}(v)$ to all neighbors of $v$
		
		\STATE receive from all neighbors of $v$ their colors $\psi_{i+1}$
		
		\ENDFOR
		
	\end{algorithmic}
\end{algorithm}

\begin{lem} \label{arbcolorparts}
	The produced coloring $\psi_{2\left \lceil \Delta/p \right \rceil + 1}$ is of the form $\langle 0,b\rangle $, $0 \leq b < q = \Theta(\Delta/p)$, for all $v \in V$.
\end{lem}
	\begin{proof}
		Consider a vertex $v \in V$. The vertex $v$ can conflict at most twice with each neighbor $u$ of different $\psi$-color within $q$ rounds, i.e., at most once before $u$ finalizes its color, and at most once after that. If $v$ conflicts with more than $p$ neighbors in each round, it means it has more than $\frac{1}{2} \cdot p \cdot (2\left \lceil \Delta/p \right \rceil + 1) > \Delta$ neighbors. This is a contradiction. Therefore, there is a round $i \in \{0,1,...,2\left \lceil \Delta/p \right \rceil\}$ in which $v$ conflicts with at most $p$ neighbors. In this round its color finalizes, i.e., becomes of the form $\langle 0,b\rangle $.
	\end{proof}

In the next lemma we bound the arbdefect of the resulting coloring. 
\begin{lem} \label{arbcol}
	The resulting coloring $\psi_{2\left \lceil \Delta/p \right \rceil + 1}$ has arbdefect at most $O(p) = O(\sqrt{\Delta})$.
\end{lem}
	\begin{proof}
		For the purpose of analysis, orient each edge $(u,v) \in E$ towards the endpoint that first set $\psi_{i+1}$ to $\langle 0,b\rangle $. If both endpoints $u,v$ did it in the same round, orient $(u,v)$ towards the endpoint with greater ID. Let $i$ denote the round in which $v$ selects a color of the form $\langle 0,b\rangle $ for the first time. Observe that once a vertex $v$ finalizes its color to $\langle 0,b\rangle $, its outgoing neighbors have already colors of the form $\langle 0,b'\rangle $.  Thus, they will never change their colors from this moment on. Moreover, the number of such neighbors of $v$ of different original $\psi$-color and the same second coordinate of $\psi_i$ is at most $p = \sqrt \Delta$. In addition, $v$ may have at most $O(p)$ neighbors with the same original $\psi$-color, since the coloring $\psi$ computed in line 1 is $O(p)$-defective. Thus, upon termination all vertices of the same $\psi_{2\left \lceil \Delta/p \right \rceil + 1}$-color induce a subgraph with arboricity $O(p)$. This is because each vertex in such a subgraph has $O(p)$ outgoing edges, each of which can be assigned a distinct label from a range of size $O(p)$. Then, all edges of the same label form a forest, and the number of forests is $O(p)$. In other words, the resulting coloring has arbdefect at most $O(p)$.
	\end{proof}

In the next lemma we analyze the running time of the algorithm.
\begin{lem} \label{arbrounds}
	The running time of the algorithm is $O(\Delta/p + \log^* n) = O(\sqrt \Delta + \log^* n)$.
\end{lem}
	\begin{proof}
		Computing a defective coloring in line 1 requires $O(\log^* n)$ time. Each iteration of the for-loop requires a single round. There are $O(\Delta/p) = O(\sqrt \Delta)$ such iterations.
	\end{proof}
The latter result gives rise to improved $(1 + \epsilon)\Delta$-coloring and $(\Delta + 1)$-coloring algorithms. This is summarized in the next theorem.

\begin{thm} \label{coloring}
	We compute $(1 + \epsilon)\Delta$-coloring within $O(\sqrt \Delta + \log^* n)$ deterministic time, for an arbitrarily small constant $\epsilon > 0$, and $(\Delta + 1)$-coloring within $O(\sqrt {\Delta \log \Delta} \log^* \Delta + \log^* n)$ deterministic time.
\end{thm}
	\begin{proof}
		In \cite{B15} it was shown that given an $O(\sqrt \Delta)$-arbdefective $O(\sqrt \Delta)$-coloring one can compute a proper  $(1 + \epsilon)\Delta$-vertex-coloring within $O(\sqrt \Delta + \log^* n)$ deterministic time. (For more details, we refer the reader to the discussion in Section 3.4 of \cite{B15}.
		However, such an arbdefective coloring is computed in \cite{B15} only within time $(\sqrt \Delta \log^3 \Delta + \log^* n)$. See Lemma 3.5, Corollary 3.12, and the discussion preceding it in \cite{B15}. Consequently, the overall running time of the algorithm of \cite{B15} for $(1 + \epsilon)\Delta$-coloring is $(\sqrt \Delta \log^3 \Delta + \log^* n)$ as well.)  Our improved running time of arbdefective coloring (cf. Lemma \ref{arbrounds}) in conjunction with the procedure of \cite{B15} (i.e., by replacing the invocation of line 1 of Algorithm 1 of \cite{B15} by an invocation of our new algorithm Arbdefective-Color), gives rise to  a deterministic $(1 + \epsilon)\Delta$-coloring within $O(\sqrt \Delta + \log^* n)$ time.

It is shown in \cite{FHK16} that a deterministic $(\Delta + 1)$-coloring is obtained in $O(\sqrt \Delta \log^{2.5} \Delta + \log^* n)$ time using arbdefective colorings. Specifically, the proof of Lemma 4.2 of \cite{FHK16} shows that  given an algorithm that starting from a proper $O(\Delta^2)$-coloring computes a $\beta$-arbdefective $k$-coloring in $O(k)$ time, then a proper $(\Delta + 1)$-coloring is computed within time $O(\log^* n + T_A)$, where $T_A$ is given by the recursive formula  $T_A(\Delta) = O(k \log^* \Delta) + T_A(O(\beta^2 \log \Delta))$. By setting $\beta = \sqrt {\Delta/ (c \log \Delta)}$ and $k = \sqrt {c \Delta \log \Delta}$, for a sufficiently large constant $c$, this recursive formula evaluates to $O(\sqrt {\Delta \log \Delta} \log^* \Delta)$. Moreover, we compute such $\beta$-arbdefective $k$-coloring within $O(\sqrt {\Delta \log \Delta} + \log^* n)$ time. (See Lemma \ref{arbrounds}.) Thus by using our Arbdefective-Color algorithm in conjunction with the procedure of \cite{FHK16}, we obtain $(\Delta + 1)$-coloring in $O(\sqrt {\Delta \log \Delta} \log^* \Delta
+ \log^* n)$ time.
\end{proof}
Hence this algorithm improves the state-of-the-art running time of $(\Delta + 1)$-coloring by a factor of $O(\log^2 \Delta / \log^* \Delta)$.

\section{3-Dimensional Additive Group Algorithm} \label{sc:3dag}

In Section \ref{sc:ag} we described our Additive Group (shortly AG) algorithm that starts from a proper $O(p^2)$-coloring, for some prime $p \geq 2 \cdot \Delta + 1$, and computes a proper $p$-coloring in $O(p)$ rounds.
This algorithm can be used, of course, also for decreasing the number of colors more than quadratically. Specifically, if we have an $O(p^3)$-coloring, for some prime $p \geq 2 \cdot \Delta +1$, we can decrease the number of colors to $O(p)$ in the following way. Partition the palette $[p^3]$ into $p$ disjoint sub-palettes $[p^2], [p^2+1, 2 p^2],...,[p^3-p^2+1,p^3]$, and run AG($p$) algorithm in each sub-palette in parallel. Within $O(p)$ rounds the number of colors reduces to $O(p^2)$, and by an additional application of AG($p$), we obtain a $p$-coloring in overall $2 \cdot O(p) =O(p)$ rounds.

In some faulty network setting it is, however, desirable to employ algorithms that do not consist of several distinct phases, like the algorithm above. These distinct phases may pose a problem when faults are introduced, and some vertices are in one phase of the algorithm, while others are in another.
We, therefore, next devise a variant of our AG algorithm that reduces the number of colors from $O(p^3)$ to $O(p)$ within $O(p)$ rounds, but it is more “uniform” than the above algorithm, i.e., at all times all vertices perform precisely the same step. We call this algorithm {\em 3-dimensional AG with a parameter $p$}, or shortly, 3AG($p$). 
The algorithm starts by representing colors $\psi(v) = \langle c_v, b_v, a_v \rangle$ as triples, $a_v,b_v,c_v \in Z_p$. 
It then runs the following iterative step for $2 \cdot p$ rounds. We will assume $p \geq 3 \cdot \Delta + 1$. All additions are in $Z_p$.

\begin{algorithm}
	\caption{3AG($p$)}
	\begin{algorithmic}[1]
		
		\FOR {$v \in V$ in parallel}
		
		\IF {$c_v \neq 0$}
		
		\IF {$\forall u \in \Gamma(v)$ it holds that $b_v \neq b_u$}
		\STATE $\psi(v) = \langle 0,b_v, a_v\rangle $
		\ELSE
		\STATE $\psi(v) = \langle c_v,b_v + c_v, a_v\rangle $
		\ENDIF
		
		\ELSE
		
		\IF {$\forall u \in \Gamma(v)$ it holds that $a_v \neq a_u$}
		\STATE $\psi(v) = \langle 0,0, a_v\rangle $
		\ELSE
		\STATE $\psi(v) = \langle 0,b_v, a_v + b_v\rangle $
		\ENDIF
		
		\ENDIF
		
		\ENDFOR
	\end{algorithmic}
\end{algorithm}

Next, we analyze the algorithm.
\begin{lem} 
	Suppose we have a proper coloring $\varphi$. Then the coloring $\psi$ obtained after one round of 3AG($p$) is proper as well.
\end{lem}
\begin{proof}
	Denote $\varphi(v)= \langle c_v,b_v,a_v \rangle$ and consider an edge $(u,v)$.
	We split the analysis into two cases, depending on whether $c_v$ is non-zero.\\
	{\bf Case 1:} $(c_v \neq 0)$. In this case our analysis splits again into two cases, depending on whether all neighbors $u'$ of $v$ have $b_u \neq b_v$, or not.\\
	{\bf Case 1.1:} ($\forall u' \in \Gamma(v)$, $b_{u'} \neq b_v$). Then the algorithm sets: $\psi(v) = \langle 0,b_v, a_v\rangle$.\\
	The vertex $u \in \Gamma(v)$ (recall that we have fixed an edge $(u,v)$) with $\varphi(u) = \langle c_u,b_u,a_u \rangle$ could have been in one of the following cases.\\
	{\bf Case 1.1.1:} $(c_u \neq 0)$. Then, if for every $z \in \Gamma(u)$, we have $b_z \neq b_u$, then $\psi(u) = \langle 0, b_u, a_u \rangle$. But recall that $b_u \neq b_v$, and thus $\psi(u) \neq \psi(v)$ as required.
	Otherwise, there exists a neighbor $z \in \Gamma(u)$ with $b_z = b_u$. Then the algorithm sets $\psi(u) = \langle c_u , b_u + c_u, a_u \rangle$
	and $c_u \neq 0$. 
	But $\psi(v) = \langle0, b_v, a_v \rangle$, i.e., $\psi(v) \neq \psi(u)$.\\
	{\bf Case 1.1.2:} $(c_u = 0)$. In this case $\varphi(u) = \langle 0, b_u, a_u \rangle$. The analysis here splits again to a number of sub-cases.\\
	{\bf Case 1.1.2.a} ($b_u = 0$). Then $\varphi(u) = \langle 0,0,a_u\rangle$, and so $\psi(u) = \langle 0,0,a_u \rangle$ as well.
	But we have for every $u' \in \Gamma(v)$, $b_{u'} \neq b_v$, and so $b_v \neq 0$. 
	Hence $\psi(v) \neq \psi(u)$.
	\\ 
	{\bf Case 1.1.2.b:} $b_u \neq 0$, but $b_u$ stayed as is, i.e.,  $\psi(u) = \langle 0, b_u, a_u+b_u \rangle$ (this means that there exists a neighbor $z \in \Gamma(u)$ with $a_z = a_u$). 
	But then again $b_v \neq b_u$, because for every $u' \in \Gamma(v)$, $b_{u'} \neq b_v$. Hence $\psi(v) \neq \psi(u)$.
	\\
	{\bf Case 1.1.2.c:} $\varphi(u) = \langle 0 , b_u, a_u \rangle$ and $b_u \neq 0$ and $\forall z \in \Gamma(u)$, $a_z \neq a_u$. Then $\psi(u) = \langle 0, 0, a_u \rangle$. But then, in particular, $a_v \neq a_u$, and so $\psi(v) = \langle 0, b_v, a_v \rangle \neq \langle 0, 0, a_u \rangle = \psi(u)$, as required.\\
	{\bf Case 1.2:} ($c_v \neq 0$, and there exists $u' \in \Gamma(v)$ with $b_{u'} = b_v$).
	Then $\psi(v) = \langle c_v,b_v+c_v, a_v \rangle$. 
	Then if $c_u = 0$ (i.e., $\varphi(u) = \langle 0, b_u,a_u \rangle$), then in $\psi(u)$ the first coordinate is also $0$  (by the rules of the algorithm), and so $\psi(v) \neq \psi(u)$.\\
	Else we have $c_u \neq 0$. So both $v$ and $u$ have non-zero first coordinate, and so they do not change their third coordinate. So if $a_v \neq a_u$ then $\psi(v) \neq \psi(u)$. 
	Otherwise ($a_v = a_u$), and so $\langle c_v,b_v \rangle \neq \langle c_u,b_u \rangle$. So if $u$ sets $\psi(u) = \langle 0,b_u,a_u \rangle$, then $\psi(v) \neq \psi(u)$, because $c_v \neq 0$.\\
	Else, $u$ sets $\psi(u) = \langle c_u,b_u + c_u,a_u \rangle$, but $\langle c_v, b_v+c_v \rangle \neq \langle c_u, b_u + c_u \rangle$ because $\langle c_v,b_v \rangle \neq  \langle c_u,b_u \rangle$. In either case $\psi(v) \neq \psi(u)$.\\
	{\bf Case 2:} ($ c_v = 0 $).
	If $\varphi(u) = \langle c_u,b_u,a_u \rangle$ and $c_u \neq 0$, then by symmetric argument, $\psi(v) \neq \psi(u)$.
	Finally, if $\varphi(v) = \langle 0,b_v,a_v \rangle$, $\varphi(u) = \langle 0,b_u,a_u \rangle$ and $\varphi(v) \neq \varphi(u)$, then by our analysis of the two-dimensional AG (see Lemma \ref{lem:propercol}), we have $\psi(v) \neq \psi(u)$.
\end{proof}
Within the first $3 \cdot \Delta + 1$ rounds, each vertex $v$ will have $c_v = 0$. This is because each neighbor $u$ of $v$ may have a conflicting $b_u$ to the $b$-value $b_v$ at most three times: once with a non-finalized $b$-value, once with a finalized $b$-value (on line 4 of the algorithm), and once with a $b$-value 0 (set on line 10 of the algorithm). So among $3 \cdot \Delta + 1$ first rounds, there will be a round on which for all $u \in \Gamma(v)$, $b_v \neq b_u$, and on that round $v$ finalizes its $b$-value (in line 4). (In fact, $2 \cdot \Delta +2$ rounds suffice, as $b_v$ can be zero at most once during all these rounds, assuming $p \geq 2 \cdot \Delta+2$.)\\
After all vertices have their $c_v = 0$, in $2 \cdot \Delta +1$ additional rounds, by the same argument, all $a_v$'s will be finalized.

\begin{col}
	The algorithm 3AG($p$), starting with a proper $p^3$-coloring, where $p \geq 2\Delta + 2$, computes a proper $p$-coloring in $O(p)$ rounds.
\end{col}

We next argue that one can decrease the palette’s size (in both ordinary and 3-dimensional variants of the algorithm AG), at the expense of slightly increasing the running time. Consider first the ordinary (two dimensional) variant of algorithm AG, and suppose that instead of running it for $p \geq 2 \cdot \Delta + 1$ rounds, we run it for $p \geq (1 + \epsilon) \cdot \Delta$ rounds, for an arbitrary small constant $\epsilon > 0$. We will run it for $1 + \lceil \frac{1}{\epsilon} \rceil$ phases, each lasting for $p$ rounds.
(Observe, however, that vertices that run the algorithm are oblivious to the phases. They always run the same AG-iteration, on which a
vertex $v$ with $\varphi(v) = \langle b_v, a_v\rangle$ checks if it has a neighbor $u$ with $a_v = a_u$. If it does not, it finalizes its color to $\psi(v) = \langle 0,a_v \rangle$. Otherwise it sets it to $\psi(v) = \langle b_v, a_v + b_v \rangle$.)
Consider a fixed vertex $v$. Note that if it does not finalize its color on phase 1, it
means that at least $\epsilon \cdot \Delta$ of its neighbors $u$ have finalized their colors (and had a conflict with the color of $v$ at least twice during the phase). Observe also that
these neighbors $u$ will be able to conflict at most once with $v$ on each subsequent phase. Hence if $v$ does not finalize its color for $i$ phases, $i  = 1,2,...,$ $i < 1/\epsilon$, it means that at least $i \cdot \epsilon \Delta$ among its neighbors did. Hence after $\lceil \frac{1}{\epsilon} \rceil$ phases, all neighbors of $v$ have
finalized their colors, and on the next phase $v$ will necessarily finalize its color.
The same reasoning is applicable to the 3-dimensional variant of the AG algorithm, but the
number of phases grows by a factor of 2.
\begin{col}\label{col:epsilon_coloring}
	Given a proper $O(p^3)$-coloring, for some $p \geq (1 + \epsilon) \cdot \Delta$, for some $\epsilon > 0$, running 3AG($p$) for $O(\frac{1}{\epsilon}\cdot p )$ rounds produces a proper $p$-coloring.
\end{col}

\section{Conclusion}
In this paper we showed that $(\Delta + 1)$-coloring can be computed using a locally-iterative algorithm below the $\Theta(\Delta \log \Delta)$ time barrier of Szegedy and Vishwanathan. In contrast to previous methods, our algorithm does not reduce the number of colors by a multiplicative factor in every single round. Instead, it guarantees that all colors enter the required range of $(\Delta + 1)$ within $O(\Delta + \log^*n)$ rounds, by performing appropriate simple operations in each round. Now, a natural question arises: is it possible to compute such a coloring using a locally-iterative algorithm with $o(\Delta) + \log^* n$ running time? While, according to previous lower bounds, this is not feasible using an algorithm that reduces the number of colors in every single iteration, a more delicate reduction with more sophisticated local rules may result in sublinear-in-$\Delta$ running time. This is a fascinating direction for future research.


\begin{thebibliography}{10}
	
	
	
	
	
	
	
	
	
	\bibitem{AAD19}
	O.~Arapoglu, V.~Akram, O.~Dagdeviren.
  \newblock An energy-efficient, self-stabilizing and distributed algorithm for maximal independent set construction in wireless sensor networks. 
	\newblock {\em Computer Standards and Interfaces}, 62: 32-42, 2019.
	
	\bibitem{B15}
	L.~Barenboim.
	\newblock Deterministic $(\Delta + 1)$-Coloring in Sublinear (in $\Delta$) Time in Static, Dynamic and Faulty Networks.
	\newblock {\em Journal of the ACM}, 63(5) : 47, 2016.
	
	
	
	\bibitem{BE09}
	L.~Barenboim, and M.~Elkin.
	\newblock Distributed $({\Delta} + 1)$- coloring in linear (in ${\Delta}$) time.
	\newblock In {\em Proc. of the 41st ACM Symp. on Theory of Computing}, pp. 111-120, 2009.
	
	\bibitem{BE10}
	L.~Barenboim, and M.~Elkin.
	\newblock Deterministic distributed vertex coloring in polylogarithmic time.
	\newblock In {\em Proc. 29th ACM Symp. on Principles of Distributed Computing}, pages 410-419,  2010.
	
	\bibitem{BE11}
	L.~Barenboim, and M.~Elkin.
	\newblock Distributed deterministic edge coloring using bounded neighborhood independence.
	\newblock In {\em Proc. of the 30th ACM Symp. on Principles of Distributed Computing}, pages 129 - 138, 2011.
	
	\bibitem{BE13}
	L.~Barenboim, and M.~Elkin.
	\newblock Distributed Graph Coloring: Fundamentals and Recent Developments.
	\newblock {\em Morgan and Claypool}, 2013.
	
	
	\bibitem{BEK14}
	L.~Barenboim, M.~ Elkin, and F.~Kuhn.
	\newblock Distributed (Delta+1)-Coloring in Linear (in Delta) Time.
	\newblock {\em SIAM Journal on Computing}, 43(1): 72-95, 2014.
	
	
	\bibitem{BEM17}
	L.~Barenboim, M.~Elkin, T.~Maimon.
	\newblock Deterministic Distributed $(\Delta + o(\Delta))$-Edge-Coloring, and Vertex-Coloring of Graphs with Bounded Diversity.
	\newblock In {\em Proc. of the 36th ACM Symp. on Principles of Distributed Computing}, pages  175-184, 2017.
	
	\bibitem{BEPS12}
	L.~Barenboim, M.~Elkin, S.~Pettie, and J.~Schneider.
	\newblock The locality of distributed symmetry breaking.
	\newblock In {\em Proc. of the 53rd Annual Symp. on Foundations of Computer Science}, pages 321-330, 2012.
	
	\bibitem{BEG18}
	L.~Barenboim, M.~Elkin U.~Goldenberg.
	\newblock Locally-Iterative Distributed ($\Delta + 1$)-Coloring below Szegedy-Vishwanathan Barrier, and Applications to Self-Stabilization and to Restricted-Bandwidth Models {\em https://arxiv.org/pdf/1712.00285.pdf}
	

\bibitem{BM12}
J.~Blair, F.~Manne.
\newblock An efficient self-stabilizing distance-2 coloring algorithm. 
\newblock {\em Theoretical Computer Science}, 444: 28-39, 2012.



	
	
	
	\bibitem{CV86}
	R.~Cole, and U.~Vishkin.
	\newblock Deterministic coin tossing with applications to optimal parallel list ranking.
	\newblock {\em Information and Control}, 70(1):32--53, 1986.
	
	\bibitem{D74} 
	E.~Dijkstra.
	\newblock Self-stabilizing systems in spite of distributed control.
	\newblock {\em Communication of the ACM}, 17 (11): 643–644, 1974.
	
	\bibitem{D00}
	S.~Dolev.
	\newblock Self-Stabilization.
	\newblock {\em MIT Press}, 2000.
	
		
	\bibitem{DH97}
	S. Dolev, and T. Herman.
	\newblock Superstabilizing Protocols for Dynamic Distributed Systems. 
	\newblock {\em Chicago J. Theor. Comput. Sci.} 1997.

	\bibitem{DGP98}
	D.~Dubhashi, D.~Grable, and A.~Panconesi.
	\newblock Nearly-optimal distributed edge-colouring via the nibble method.
	\newblock {\em Theoretical Computer Science, a special issue for the best papers of ESA95}, 203(2):225--251, 1998.
	
	
	
	
	
	\bibitem{EPS15}
	M.~Elkin, S.~Pettie, and H.~Su.
	\newblock $(2\Delta - 1)$-Edge-Coloring is Much Easier than Maximal Matching in the Distributed Setting.
	\newblock In {\em Proc. of the 26th ACM-SIAM Symp. on Discrete Algorithms}, pages 355-370, 2015.
	
	
	
	
	\bibitem{FGK17}
	M.~Fischer, M.~Ghaffari, and F.~Kuhn .
	\newblock Deterministic Distributed Edge Coloring via Hypergraph Maximal Matching.
	\newblock To appear in {\em 58th Annual Symp on Foundations of Computer Science}, 2017.
	
	\bibitem{FHK16}
	P.~Fraigniaud, M.~Heinrich, and A.~Kosowski.
	\newblock Local Conflict Coloring.
	\newblock In {\em Proc. of the 57th Annual Symp. on Foundations of Computer Science}, pages 625 - 634, 2016.
	
	
	
	
	\bibitem{GP97}
	D.~Grable, and A.~Panconesi.
	\newblock Nearly optimal distributed edge colouring in O(log log n) rounds.
	\newblock {\em Random Structures and Algorithms}, 10(3): 385-405, 1997.
	
	
	
	
	\bibitem{GP87}
	A.~Goldberg, and S.~Plotkin.
	\newblock Parallel ($\Delta+1$)-Coloring of Constant-Degree Graphs.
	\newblock {\em Inf. Process. Lett.} 25(4): 241-245, 1987.
	
	\bibitem{GPS88}
	A.~Goldberg, S.~Plotkin, and  G.~Shannon.
	\newblock Parallel symmetry-breaking in sparse graphs.
	\newblock {\em SIAM Journal on Discrete Mathematics}, 1(4):434--446, 1988.
	
	\bibitem{GK10}
	N.~Guellati, and H.~Kheddouci.
	\newblock A survey on self-stabilizing algorithms for independence, domination, coloring, and matching in graphs.
	\newblock {\em Journal of Parallel and Distributed Computing}, 70(4): 406-415, 2010.
	
	
	
	\bibitem{hardy_write}
	G.~Hardy, and E.~Wright.
	\newblock An introduction to the theory of numbers.
	\newblock {\em Oxford university press}, 5th edition, 1980.
	
	
	
	\bibitem{HKMS15}
	D.~Hefetz, F.~Kuhn, Y.~Maus, and A.~Steger.
	Polynomial Lower Bound for Distributed Graph Coloring in a Weak LOCAL Model. 
	\newblock In {\em Proc. of the 30th International Symp. on DiStributed Computing}, pages 99 - 113, 2016. 
	
	\bibitem{H02}
	T.~Herman.
	\newblock Self-stabilization bibliography: Access guide.
	\newblock {\em Chicago Journal of Theoretical Computer Science},  Working Paper WP-1, 2002.
	
	
	\bibitem{HH92}
	S.C.~Hsu, S.T.~Huang.
	\newblock A self-stabilizing algorithm for maximal matching.
	\newblock {\em Information Processing Letters}, 43 (2):7781, 1992 .
	
	
	
	
	\bibitem{IKK02}
	M.~Ikeda, S.~Kamei, and H.~Kakugawa.
	\newblock A space-optimal self-stabilizing algorithm for the maximal independent set problem.
	\newblock In {\em Proc. 3rd International Conference on Parallel and Distributed Computing, Applications and Technologies}, 2002.
	
	
	\bibitem{KKT15}
	V.~King, S.~Kutten, M.~Thorup.
\newblock Construction and Impromptu Repair of an MST in a Distributed Network with o(m) Communication.
\newblock In {\em Proc. of the 34th ACM Symp. on Principles of Distributed Computing}, pages 71-80, 2015.
	
	
	\bibitem{KK06}
	A. Kosowski, L. Kuszner.
	\newblock Self-stabilizing algorithms for graph coloring with improved performance guarantees.
	\newblock In {\em Proc. 8th International Conference on Artificial Intelligence and Soft Computing}, pages 1150-1159 2006.
	
	
	
	
	
	\bibitem{KSOS06}
	K.~Kothapalli, C.~Scheideler, M.~Onus, and C.~Schindelhauer.
	\newblock Distributed coloring in ~O($\sqrt{\log n}$) bit rounds.
	\newblock In {\em Proc. of the 20th International Parallel and Distributed Processing Symp.}, 2006.
	
	\bibitem{K09}
	F.~Kuhn.
	\newblock Weak graph colorings: distributed algorithms and applications. 
	\newblock In {\em Proc. of the 21st ACM Symp. on Parallel Algorithms and Architectures}, pages 138--144, 2009.
	
	
	
	
	
	
	
	\bibitem{KW06}
	F.~Kuhn, and R.~Wattenhofer.
	\newblock On the complexity of distributed graph coloring.
	\newblock In {\em Proc. 25th ACM Symp. Principles of Distributed Computing}, pp. {7--15}, 2006.
	
\bibitem{LL14}
	C.~Lee, and T.~Liu.
  \newblock A Self-Stabilizing Distance-2 Edge Coloring Algorithm. 
	\newblock {\em The Computer Journal}, 57(11): 1639-1648, 2014.
	
	
	\bibitem{LSW09}
	c.~Lenzen, J.~Suomela, and R.~Wattenhofer.
	\newblock  Local algorithms: Self-stabilization on speed. 
	\newblock In {\em Proc. of the 11th Symposium on Self-Stabilizing Systems}, pp. 17-34, 2009.
	
	
	
	\bibitem{L87}
	N.~Linial.
	\newblock Distributive graph algorithms: Global solutions from local data
	\newblock In {\em Proc. 28th Symp. on Foundation of Computer Science}, pp. {331--335}, 1987.
	
	
	\bibitem{L66}
	L.~Lovasz. 
	\newblock On decompositions of graphs.
	\newblock {\em Studia Sci. Math. Hungar.}, 1:237–238, 1966.
	
	
	
	
	\bibitem{NS93}
	M.~Naor, and L.~Stockmeyer.
	\newblock What can be computed locally?
	\newblock In {\em Proc. 25th ACM Symp. on Theory of Computing}, pages 184-193, 1993.
	
	\bibitem{PPPRR17}
	S.~Pai, G.~Pandurangan, S.~Pemmaraju, T.~Riaz, and P.~Robinson.
	\newblock Symmetry Breaking in the Congest Model: Time- and Message-Efficient Algorithms for Ruling Sets.
	\newblock https://arxiv.org/abs/1705.07861
	
	
	\bibitem{PR01}
	A.~Panconesi, and R.~Rizzi.
	\newblock Some simple distributed algorithms for sparse networks.
	\newblock {\em Distributed Computing}, 14(2):97--100, 2001.
	
	
	\bibitem{PS97}
	A.~Panconesi, and A.~Srinivasan.
	\newblock Randomized Distributed Edge Coloring via an Extension of the Chernoff-Hoeffding Bounds.
	\newblock {\em SIAM Journal on Computing}, 26(2):350-368, 1997.
	
	\bibitem{P00}
	D.~Peleg.
	\newblock {\em Distributed Computing: A Locality-Sensitive Approach.}
	\newblock SIAM, 2000.
	
	
	
	\bibitem{SS93}
	S.~Sur, and P.K.~Srimani.
	\newblock A self-stabilizing algorithm for coloring bipartite graphs.
	\newblock {\em Information Sciences}, 69, pages 219-227, 1993 .
	
	\bibitem{SV93}
	M.~Szegedy, and S.~Vishwanathan.
	\newblock Locality based graph coloring.
	\newblock In {\em Proc. 25th ACM Symp. on Theory of Computing}, pages 201-207, 1993.
	
	
	
	\bibitem{V64}
	V.~Vizing.
	\newblock On an estimate of the chromatic class of a p-graph.
	\newblock {\em Metody Diskret. Analiz}, 3: 25-30, 1964.

	
	
	
	
	
	%
	
	
	
	
	
	
	
	%
	
\end{thebibliography}
\end{document}